\documentclass[11pt]{article}
\usepackage{amssymb,latexsym,amsmath}
\usepackage{xcolor}
\usepackage{bm}
\usepackage{amsmath,amscd}
\usepackage{graphicx}
\usepackage{newtxtext}
\usepackage{newtxmath}
\usepackage{geometry}
\usepackage{natbib}
\usepackage{hyperref}
\hypersetup{
    colorlinks = true,
    urlcolor   = blue,
    citecolor  = black,
}
\newtheorem{theorem}{Theorem}[section]

\newtheorem{definition}{Definition}[section]
\newcommand{\RomanNumeralCaps}[1]
\linenumbers

\title{On Jean-Marie Souriau's\\
 geometric quantization of the relativistic electron}

\author{\textbf{G. de Saxc\'e} \\
Univ. Lille, CNRS, Centrale Lille, UMR 9013 – LaMcube \\
Laboratoire de m\'ecanique multiphysique multi\'echelle, \\
F-59000, Lille, France, Email: gery.de-saxce@univ-lille.fr}

\begin{document}
\maketitle
\begin{abstract}
The aim of this paper is to revisit, in Souriau's book "Structure of Dynamical Systems", the chapter devoted to the geometric quantization where the justifications of important results and formulae are not given and are difficult to prove. After recalling the coadjoint orbit method and its application to the relativistic particle with spin, we state and prove two keystone theorems that allow to equip its  prequantum manifold with the symplectic and contact structures. We apply them to the relativistic particle with spin 1/2, leading to the Dirac equation and the conservation of the probability current. We propose also an identity of conservation of the spin current and, invoking the Kaluza-Klein theory, a systematic construction of the symmetries of charge conjugation, parity transformation and time reversal which seems to us more convenient and readable than that of the classical presentations.
\end{abstract}           


\vspace{0.2cm}










\section{Introduction}

In a seminal work,  Jean-Marie Souriau proposed for the classical description of an elementary particle to study a family of homogeneous spaces for a symmetry group, the coadjoint orbits, and to classify them (\cite{SSD, SSDEng}). On this basis, he introduced in \cite{Souriau 1966} the geometric quantization which is  a refinement of the representation theory based on the classical-quantum correspondance. The reader not  familiar with this topic can consult, in addition to Souriau's references previously cited, \cite{Kostant 1970} who is with Souriau the other pioneer of this method, \cite{Kijowski 1977} and \cite{Woodhouse 1991}. The pre-quantization step consists in constructing a  $\mathsf{U}(1)$-principal bundle of which the base space is a symplectic manifold, a coadjoint orbit of the symmetry group. By classifying coadjoint orbits, one classifies the elementary particles known to physicists, whether elementary such as quarks and leptons or composite such as hadrons. On this basis, he proposed an intrinsic  method of quantization  that gives a precise sense to the correspondence principle.

As he says in the introduction of "Structure of Dynamical Systems" (\cite{SSD, SSDEng}), this method "leaves open many mathematical problems and questions of interpretation. We present this method of quantization above all as a subject of reflection for the researcher". Our goal in this work is to revisit his quantization of the relativistic particle with spin $\frac{1}{2}$ leading to Dirac equation for the electron. The justifications of important results and formulae are not given and are difficult to prove. For example, the proof of (\cite{SSD, SSDEng}, Formula (18.71)) seems requiring the spin group which is not considered in this book. On the other hand, Souriau developped an original theory of spinors with a purely algebraic construction of the spin group in "Calcul lin\'eaire" (\cite{CL}) and in "G\'eom\'etrie et Relativit\'e" (\cite{GR}). This theory seemed to us also worthy of interest and to be made available to the specialists in an English version. It will be our starting point. Although the two approaches, that of "Structure of Dynamical Systems" on the one hand and that of  "Calcul lin\'eaire" and "G\'eom\'etrie et Relativit\'e" on the other, are consistent each of them, it must be stressed, however, that the Dirac matrices of the former one are biquaternionic (\cite{Breban 2018}) while that of the last one are quaternionic. After presenting the aforementioned spinor theory, we shall developped the electron quantization with quaternionic Dirac matrices for consistency although following the frame of "Structure of Dynamical Systems". Another difference with this book is, in the axioms of the prequantum manifold, the re-introduction as in \cite{Souriau 1975} and \cite{Duval 2011} of the Sommerfeld quantization condition, only referred in  (\cite{SSD, SSDEng}, (18.6)) but not used and replaced by (18.19). Many tools developped in this paper are not restricted to the relativistic particle with mass and spin but can be useful to study other kinds of particles, for instance the neutrinos (see \cite{SSD, SSDEng}, (19.122) to (19.134)).

This paper is a preliminary work of the author needed to tackle the four dimensional Dirac equation in the footstep of a previous article  (\cite{de Saxce 2025}) where he proposed a new symmetry group and a modified version of Kaluza-Klein theory in which a variational principle allows to find new coupled equations of the gravitation and electromagnetic field. After considering the 5D approach at the Universe scale, the aim will be to zoom in to understand with this new point of view what happens at the quantum scale. 

The paper is structured as follows.

In Section \ref{Section Generalized Hermitian spaces}, we recall a generalization of classical Hermitian spaces with a metric which is not necessarily positive definite, as proposed in \cite{CL}. Section \ref{Section - Dirac spinors} is devoted to a formulation of a spinor theory based on the use of a space of quaternionic matrices and the corresponding canonical decomposition. In Section \ref{Section - Spin group}, we construct explicitly the spin group in an algebraic manner, first in the formulation linked to the 5D universe, next in the reduced one corresponding to the classical space-time with the Minkowski metric. Section \ref{Section - Vector product} is a reminder of the generalization of the classical vector product in arbitrary finite dimension, as in \cite{CL}. Section \ref{Section - The coadjoint orbit method} is a synthesis of the formulation, as developed in \cite{SSD}, of the coadjoint orbit method and its application to the relativistic particle with spin. Section \ref{Section - Geometric prequantization} is a presentation of the geometric prequantization as proposed in \cite{SSD} with the variant using the Sommerfeld quantization condition as in \cite{Souriau 1975} and \cite{Duval 2011}. In Section \ref{Section - Prequantization of the relativistic particle with  spin 1/2}, we prequantize the space of motions of the relativistic particle with spin described in Section \ref{Section - The coadjoint orbit method}, using the decomposition of Section \ref{Section - Dirac spinors}. Section \ref{Section - Symplectic and contact structures of the prequantum manifold} contains the two keystone theorems that allow to equip the prequantum manifold with the symplectic and contact structures. Section \ref{Section - Geometric quantization} is a reminder of the procedure of geometric quantization as developed in \cite{SSD}. It is applied in Section \ref{Section - Quantization of the relativistic particle with  spin 1/2} to the relativistic particle with  spin 1/2, leading to the Dirac equation and the conservation of the probability current. We propose also an identity of  conservation of the polarization current or the spin current. In \ref{Section - Charge conjugation, parity transformation and time reversal symmetries}, invoking the Kaluza-Klein theory, we present a systematic construction of the symmetries of charge conjugation, parity transformation and time reversal  which seems to us simpler and more readable than that of the classical presentations.

\section{Generalized Hermitian spaces}
\label{Section Generalized Hermitian spaces} 

Let us begin with a brief reminder of definitions and notations used by Souriau and which stray from standard ones. Their interest is to ensure the frame independent, intrinsic feature of the approach, even if calculations are often made in a particular basis. 

In this paper, an \textbf{Euclidean space}  (called by other authors a pseudo-Euclidean space) is a real vector space  $\mathcal{T}$ of finite dimension $m$ equipped with a (covariant) metric $\bm{G}$, {\it i.e.} a nondegenerate symmetric 2-covariant tensor $\bm{G}$. If $\bm{X}, \bm{X}' \in \mathcal{T}$ and $\bm{G}$ are represented in a basis respectively by the columns $X, X'$ and the matrix $G$, then
$$ \bm{G} (\bm{X}, \bm{X}') = X^T G\,X'
$$ 
If $G$ is diagonal with element $+1$ or $-1$ on the diagonal, the basis is called orthonormal. The number $p$ of positive values on the diagonal is called the positive index of inertia. If $m  \geq 2$ and $p = 1$, the space is said hyperbolic. To every vector $\bm{X}$ is associated one and only one linear form $\bm{X}^*$  such that $ \bm{X}^* \bm{X}'  = \bm{X}^* (\bm{X}') = \bm{G} (\bm{X}, \bm{X}')$, called its \textbf{adjoint} and represented by the row $ X^* = X^T G$ . More generally, if $\bm{A}$ is a linear map from an Euclidean space $\mathcal{T}_0$ into another one $\mathcal{T}$, its adjoint is the linear map $\bm{A}^*$ from $\mathcal{T}$ into $\mathcal{T}_0$ such that
$$ \forall \bm{X}\in\mathcal{T},\quad \forall \bm{X}_0\in\mathcal{T}_0,\qquad 
    \bm{G} (\bm{X}, \bm{A}\,\bm{X}_0) = \bm{G}_0 (\bm{A}^*\bm{X},  \bm{X}_0)\
$$ 
If $\bm{A}$ is represented by the matrix $A$ in bases of $\mathcal{T}_0$ and $\mathcal{T}$, $\bm{A}^*$  is represented by the matrix
\begin{equation}
   A^* = G_0^{-1} A^T G\ .
\label{A^* = G_0^(-1) A^T G} 
\end{equation}
We verify that:
\begin{equation}
 (\bm{A} + \bm{A}')^* = \bm{A}^* + \bm{A}'^*,\qquad 
    (\bm{A}\,\bm{B})^* = \bm{B}^* \bm{A}^*,\qquad 
    (\bm{A}^*)^* = \bm{A}
\label{rules for the adjoint}
\end{equation}
The linear map is self-adjoint (resp. anti-self-adjoint or skew-adjoint) if
$$ \bm{A} = \bm{A}^*\qquad (\mbox{resp.}\quad \bm{A} = - \bm{A}^*)\ .
$$

If $\mathcal{T}^\mathbb{C}$ is a complexification of a $n$-dimensional Euclidean real vector space $\mathcal{T}$, it is, as $2 \, n$-dimensional real vector space, an Euclidean space for
$$ \forall \bm{X}, \bm{X}',  \bm{Y}, \bm{Y}' \in \mathcal{T}, \qquad
   \bm{G}^\mathbb{C} (\bm{Z}, \bm{Z}') = \bm{G}^\mathbb{C} (\bm{X} + i \, \bm{Y}, \bm{X}' + i \, \bm{Y}')  
                                          = \bm{G} (\bm{X}, \bm{X}')  + \bm{G} (\bm{Y}, \bm{Y}') 
$$
We say it is a \textbf{Hermitian space}  (called "Euclidean space" in \cite{SSD, SSDEng}).
As $\bm{G}^\mathbb{C} (\bm{Z}, \bm{Z}') $ must be real, we define the Hermitian \textbf{adjoint}  of $\bm{Z}$ (called "transpose" by Souriau) as the unique $\mathbb{C}$-linear form $\bm{Z}^\dag$ such that 
\begin{equation}
    \bm{G}^\mathbb{C} (\bm{Z}, \bm{Z}')  = \Re (\bm{Z}^\dag  \bm{Z}')  
\label{G (Z, Z') = Re (Z^dag Z')}
\end{equation}
$\mathbb{C}$ is a Hermitian space for the metric
$$ \bm{G}^\mathbb{C}(z, z') = \bm{G}^\mathbb{C} (x + i \, y, x' + i \, y')  = x \, x' + y  \, y'  = \Re (\bar{z} \, z')
$$
then, owing to (\ref{G (Z, Z') = Re (Z^dag Z')}), the adjoint $z^\dag$ of $z$ is its conjugate $\bar{z}$. If $\bm{Z}$ is represented in a $\mathbb{C}$-basis of $\mathcal{T}^\mathbb{C}$ by the column $Z$, its adjoint $\bm{Z}^\dag$ is represented by the row 
$$ Z^\dag = \bar{Z}^T G
$$
The adjoint of a $\mathbb{C}$-linear map $\bm{A}$ between Hermitian spaces will be denoted $\bm{A}^\dag$. If $\bm{A}$ is represented by the matrix $A$ in $\mathbb{C}$-bases, $\bm{A}^\dag$  is represented by the matrix
\begin{equation}
   A^\dag = G_0^{-1} \bar{A}^T G
\label{A^dag = G_0^(-1) bar(A)^T G} 
\end{equation}
The rules concerning the adjoint for Hermitian spaces are similar to the corresponding ones for Euclidean spaces, replacing $*$ by $\dag$ in (\ref{rules for the adjoint}). $\Re (\bm{Z}^\dag  \bm{Z}')  $ is symmetric with respect to $Z$ and $Z'$ but it is worth to note that $\Im (\bm{Z}^\dag  \bm{Z}')  $ is skew-symmetric with respect to $Z$ and $Z'$
$$ \Im (\bm{Z}^\dag  \bm{Z}')  = - \Im (\overline{\bm{Z}^\dag  \bm{Z}'})  
= - \Im ((\bm{Z}^\dag  \bm{Z}')^\dag) 
$$
\begin{equation}
   \Im (\bm{Z}^\dag  \bm{Z}') = - \Im (\bm{Z}'^\dag  \bm{Z}) 
\label{Im (Z^dag Z') - - Im(Z'^dag Z)}
\end{equation}
A linear map $\bm{A} $ is Hermitian (resp. anti-Hermitian) if:
$$ \bm{A} = \bm{A}^\dag\qquad (\mbox{resp.}\quad \bm{A} = - \bm{A}^\dag)
$$
In particular, if $\mathbb{R}^n$ is equipped with the standard Euclidean metric, the adjoint of $Z \in \mathbb{C}^n$ is the row $Z^\dag = \bar{Z}^T$ and the adjoint of the matrix $A \in \mathbb{C}^{n \times n}$ is the matrix $A^\dag = \bar{A}^T$. Hence we recover the classical definition of Hermitian spaces. The previous considerations provide a generalization of classical Hermitian spaces with a metric which is not necessarily positive definite and that will allow to work with Dirac matrices which are in full coherence all Hermitian. 

\section{Dirac spinors}
\label{Section - Dirac spinors}

Throughout the paper, we are working in Planck natural units for which $\hslash = c = 1$. The spinors are elements $\bm{\psi}$ of a Hermitian space $\mathcal{E}$ of dimension $4$ and index of inertia $2$. In a suitable orthonormal basis, it is represented by the matrix
\begin{equation}
     G = \left[ {{\begin{array}{cc}
       1_{\mathbb{R}^2}    \hfill &    0                           \hfill  \\
        0                            \hfill &  - 1_{\mathbb{R}^2} \hfill  \\
   \end{array} }} \right] 
\label{G = ((1 0)(0 -1))}   
\end{equation}
In this basis, the spinor $\bm{\psi}$ is represented by the column
\begin{equation}
    \psi = \left[ {{\begin{array}{c}
           \psi'                           \hfill  \\
           \psi''                           \hfill  \\
   \end{array} }} \right] 
\label{psi = (psi' psi'')}
\end{equation}
where $\psi'$ and $\psi''$ are called half-spinors
$$   \psi' = \left[ {{\begin{array}{c}
           z'_1                           \hfill  \\
           z'_2                           \hfill  \\
   \end{array} }} \right] , \qquad
      \psi'' = \left[ {{\begin{array}{c}
           z''_1                           \hfill  \\
           z''_2                           \hfill  \\
   \end{array} }} \right] 
$$
They are elements of the Hermitian space $\mathbb{C}^2$ with the standard metric, then 
$$   \psi'^\dag  \psi' = \left[ \overline{z'_1} \; \; \overline{z'_2}\right] \,
\left[ {{\begin{array}{c}
           z'_1                           \hfill  \\
           z'_2                           \hfill  \\
   \end{array} }} \right] = \mid  z'_1 \mid^2 + \mid  z'_2 \mid^2 = \mid \psi' \mid^2
$$
and similar notations for $\psi''$. Consequently, one has
$$   \psi^\dag  \psi = \left[ \psi'^\dag \; \; - \psi''^\dag\right] \,
\left[ {{\begin{array}{c}
           \psi'                           \hfill  \\
           \psi''                           \hfill  \\
   \end{array} }} \right] =  \mid  \psi' \mid^2 - \mid  \psi'' \mid^2 
$$
We represent the quaternions as the elements of the real space $\mathbb{H}$ of matrices of $\mathbb{C}^{2 \times 2}$ spanned by the $\mathbb{R}$-basis
$$ 1_{\mathbb{C}^2} = \left[ {{\begin{array}{cc}
       1    \hfill &    0                           \hfill  \\
        0                            \hfill &  1    \hfill  \\
   \end{array} }} \right] , \qquad
   I = \left[ {{\begin{array}{cc}
       0    \hfill &    - i                           \hfill  \\
        - i                            \hfill &  0    \hfill  \\
   \end{array} }} \right] , \qquad
      J = \left[ {{\begin{array}{cc}
       0    \hfill &    - 1                          \hfill  \\
        1                            \hfill &  0    \hfill  \\
   \end{array} }} \right] , \qquad
      K = \left[ {{\begin{array}{cc}
       - i    \hfill &    0                          \hfill  \\
        0                            \hfill &  i    \hfill  \\
   \end{array} }} \right] 
$$
such that $ I^2 =J^2 =K^2 = - 1_{\mathbb{C}^2}, \quad 
J K = - K J = I, \quad 
K I = - I K = J, \quad 
I J = - J I = K $ and the \textbf{Pauli matrices} are 
$$ \sigma_1 = i \, I, \quad \sigma_2 = i \, J \quad 
\sigma_3 = i \, K
$$ 
We have the decomposition
$$  Q \in \mathbb{H} \qquad \Leftrightarrow \qquad 
      Q = y \, 1_{\mathbb{C}^2}  + x_1 I + x_2 J + x_3 K, \quad y, x_1, x_2, x_3 \in \mathbb{R}
$$
then its adjoint $Q^\dag$ is it conjugate $\bar{Q}$
$$ Q^\dag = \bar{Q}  = y \, 1_{\mathbb{C}^2}  - x_1 I - x_2 J - x_3 K
$$

Let us consider the vector space $End(\mathcal{E})$ of linear maps from $\mathcal{E}$ into itself and the subspace $\bm{\Gamma}$ of its elements  $\bm{\gamma}$  which are represented in a suitable orthonormal basis of $\mathcal{E}$ by quaternionic matrices, {\it i.e.} the elements of $\mathbb{C}^{4 \times 4}$
$$       \gamma= \left[ {{\begin{array}{cc}
       A  \hfill &  C    \hfill  \\
       B  \hfill &  D    \hfill  \\
   \end{array} }} \right] 
$$
of which the $2 \times 2$ blocks $A, B, C, D$ are quaternions. The set $\Gamma$ of quaternionic matrices is a real vector space of dimension 16 that admits a \textbf{canonical decomposition} 
\begin{equation}
   \Gamma = \Gamma_i \oplus \Gamma_{+} \oplus \Gamma_{-}
\label{Gamma = Gamma-i oplus Gamma_(+) oplus Gamma_(-)}
\end{equation}
into the set of isotropic matrices 
$$ \Gamma_i  = \lbrace \gamma =  \lambda \, 1_{\mathbb{C}^4}, \quad \lambda \in \mathbb{R} \rbrace 
$$
the set of Hermitian traceless matrices
$$ \Gamma_{+}  = \lbrace \gamma \in \mathcal{M} \quad \mbox{such that} \quad Tr (\gamma) = 0, \quad \gamma^\dag = \gamma \rbrace 
$$
and the set of anti-Hermitian traceless matrices
$$ \Gamma_{-}  = \lbrace \gamma \in \mathcal{M} \quad \mbox{such that} \quad Tr (\gamma) = 0, \quad \gamma^\dag = - \gamma \rbrace 
$$
Taking into account (\ref{A^dag = G_0^(-1) bar(A)^T G}) and (\ref{G = ((1 0)(0 -1))}), the adjoint of quaternionic matrices is
$$       \gamma= \left[ {{\begin{array}{cc}
       \bar{A}  \hfill &  - \bar{B}    \hfill  \\
       - \bar{C}  \hfill &  \bar{D}    \hfill  \\
   \end{array} }} \right] 
$$
If $\gamma \in \Gamma_{+}$, we find $A = \bar{A}$, $D = \bar{D}$, $B = - \bar{C}$ and  $Tr (A) + Tr (D) = 0$.  Consequently, there exists $t \in \mathbb{R}$ such that $A = t \, 1_{\mathbb{C}^2}$, $D = - t \, 1_{\mathbb{C}^2} $ and finally
\begin{equation}
       \gamma \in  \Gamma_{+} \;\; \Leftrightarrow \;\;
        \gamma= \left[ {{\begin{array}{cc}
        t \, 1_{\mathbb{C}^2}  \hfill &  - \bar{B}  \hfill  \\
       B  \hfill &  -  t \, 1_{\mathbb{C}^2}    \hfill  \\
   \end{array} }} \right]  
   = \left[ {{\begin{array}{cc}
        t \, 1_{\mathbb{C}^2}                                        \hfill &  - y \, 1_{\mathbb{C}^2}  + x_1 I + x_2 J + x_3 K  \hfill  \\
       y \, 1_{\mathbb{C}^2}  + x_1 I + x_2 J + x_3 K  \hfill &  -  t \, 1_{\mathbb{C}^2}                                         \hfill  \\
   \end{array} }} \right]
\label{gamma in Gamma_(+) <=>}
\end{equation}
where $t, x_1, x_2, x_3, y \in \mathbb{R}$. Therefore $\Gamma_{+}$ is a $5$-dimensional subspace of $ \Gamma$ because spanned by the basis
$$ \gamma_0 = \left[ {{\begin{array}{cc}
        1_{\mathbb{C}^2}  \hfill &  0 \hfill  \\
       0  \hfill &  -  1_{\mathbb{C}^2}    \hfill  \\
   \end{array} }} \right], \quad
      \gamma_5 = \left[ {{\begin{array}{cc}
        0  \hfill &  -1_{\mathbb{C}^2} \hfill  \\
       1_{\mathbb{C}^2} \hfill &  0    \hfill  \\
   \end{array} }} \right]
 $$
 \begin{equation}
  \gamma_1 = \left[ {{\begin{array}{cc}
        0  \hfill &  I \hfill  \\
       I  \hfill &  0    \hfill  \\
   \end{array} }} \right], \quad
   \gamma_2 = \left[ {{\begin{array}{cc}
        0  \hfill &  J \hfill  \\
       J  \hfill &  0    \hfill  \\
   \end{array} }} \right], \quad
   \gamma_3 = \left[ {{\begin{array}{cc}
        0  \hfill &  K \hfill  \\
       K  \hfill &  0    \hfill  \\
   \end{array} }} \right]
 \label{defi gamma_alpha}
 \end{equation}
Likewise, we can prove that 
$$    \gamma \in  \Gamma_{-} \;\; \Leftrightarrow \;\;
        \gamma = \left[ {{\begin{array}{cc}
         x_1 I + x_2 J + x_3 K                                      \hfill &  - y \, 1_{\mathbb{C}^2}  - x''_1 I - x''_2 J - x''_3 K  \hfill  \\
       - y \, 1_{\mathbb{C}^2}  + x''_1 I + x''_2 J + x''_3 K  \hfill &   x'_1 I + x'_2 J + x'_3 K                                   \hfill  \\
   \end{array} }} \right]
$$
Therefore $ \Gamma_{-}$ is a $10$-dimensional subspace of $ \Gamma$ because spanned by the basis composed of the six matrices $\gamma_i \gamma_j$ with $0 \leq i < j \leq 3$ and the 4 matrices $\gamma_i \gamma_5$ with $0 \leq i \leq 3$. 

\begin{center}
\begin{tabular}{l l l}
 \hline \hline
  space                                                                                             & elements                             & dimension  \\
  \hline \hline
  $\Gamma = \Gamma_i \oplus \Gamma_{+} \oplus \Gamma_{-}$  & quaternionic                        & 16              \\
    \hline
  $\Gamma_i$                                                                                  & isotropic                              & 1                 \\
  $\Gamma_{+}$                                                                              & Hermitian and traceless       & 5                  \\
  $\Gamma_{-}$                                                                              & anti-Hermitian and traceless & 10                 \\
   \hline \hline
\end{tabular}
\end{center}
\begin{center}
\textbf{Table 1:} Quaternionic matrix space and subspaces
\end{center}

The decomposition is summarized in Table 1.
Let us focus now our attention on the properties of the space $\Gamma_{+}$. By convention, The Greek indices are $0, 1, 2 , 3, 5$ while the Latin indices are $0, 1, 2 , 3$. We define the map
\begin{equation}
    \gamma: \mathbb{R}^5 \rightarrow \Gamma_{+}: X \mapsto \gamma (X) = X^\alpha \gamma_\alpha   
\label{gamma from R^5 into Gamma_(+)}
\end{equation}
Comparing to (\ref{gamma in Gamma_(+) <=>}) allows us to do the identification
$$ X^0 = t,\;\; 
     X^1 = x_1,\;\, 
     X^2 = x_2,\;\,
     X^3 = x_3,\;\,
     X^5 = y
$$
This shows that the map $\gamma$ is $\mathbb{R}$-linear and one-to-one between $\mathbb{R}^5$ and $\Gamma_{+}$. Next, we have 
$$       (\gamma (X))^2 = \left[ {{\begin{array}{cc}
        t \, 1_{\mathbb{C}^2}  \hfill &  - \bar{B}  \hfill  \\
       B  \hfill &  -  t \, 1_{\mathbb{C}^2}    \hfill  \\
   \end{array} }} \right]^2 
   = \left[ {{\begin{array}{cc}
        t^2 1_{\mathbb{C}^2} - \bar{B} \, B \hfill &  0  \hfill  \\
       0  \hfill &  - B \, \bar{B}  +  t^2 1_{\mathbb{C}^2}    \hfill  \\
   \end{array} }} \right]
$$ 
$$       (\gamma (X))^2 =  \lbrack t^2 - (x_1)^2 - (x_2)^2- (x_3)^2- y^2 \rbrack \, 1_{\mathbb{C}^4} 
$$ 
If we equip $E_{1, 4} = \mathbb{R}^5$ with a structure of a hyperbolic Euclidean space for the metric of signature $1 + 4$
\begin{equation}
      X^* X =  (X^0)^2 - (X^1)^2 - (X^2)^2 -(X^3)^2 - (X^5)^2
               =  t^2 - (x_1)^2 - (x_2)^2- (x_3)^2- y^2
\label{X^* X = metric of E_(1,4)}
\end{equation}
we obtain
\begin{equation}
     (\gamma (X))^2 =  (X^* X)  \, 1_{\mathbb{C}^4} 
\label{(gamma(X))^2 = (X^* X ) 1_(C^4)}
\end{equation}
from which we deduce easily 
\begin{equation}
   \forall U, V \in \mathbb{R}^5, \qquad 
   \frac{1}{2} \lbrack \gamma(U) \gamma(V) + \gamma(V) \gamma(U) \rbrack
    =  (U^* V)  \, 1_{\mathbb{C}^4} 
\label{gamma(U) gamma(V) + gamma(V) gamma(U) = (U^* V) 1_(C^4)}
\end{equation}
In particular, we have
\begin{equation}
       (\gamma_0)^2 = 1_{\mathbb{C}^4} , \quad
     (\gamma_1)^2 = (\gamma_2)^2 = (\gamma_3)^2 = (\gamma_5)^2 =  - 1_{\mathbb{C}^4}
\label{(gamma_alpha)^2 = for alpha = 0...5}
\end{equation}
\begin{equation}
      \gamma_\alpha \gamma_\beta + \gamma_\beta \gamma_\alpha = 0 \quad \mbox{if} \quad \alpha \neq \beta
\label{gamma_alpha gamma_beta + gamma_beta gamma_alpha = 0 for alpha neq beta}
\end{equation}
As $\gamma(U)$ belongs to $\Gamma_{+}$, it is Hermitian then, for all spinor $\psi$, the scalar $\psi^\dag \gamma(U) \, \psi$ is real. For a timelike vector $U$ of  $E_{1, 4} $ ($U^* U > 0$), we have
\begin{equation}
\psi^\dag \gamma(U) \, \psi 
\begin{cases}
> 0 & \mbox{if} \; U \; \mbox{ is future-directed}\\
< 0 & \mbox{if} \; U \; \mbox{ is past-directed}
\end{cases}
\label{psi^dag gamma(U) psi > 0 if U is future-directed}
\end{equation}
For the proof, the reader is referred to (\cite{CL}, \S 31, (19)).

\section{Spin group}
\label{Section - Spin group}

The  set of matrices $S$ of $\Gamma$ which are unitary
$$ S^\dag S = 1_{\mathbb{C}^4}
$$
is a group denoted $\mbox{Spin}(1,4)$ and called the \textbf{spin group}. 

\begin{theorem}[Souriau]
The elements of the spin group have the following properties:
\begin{itemize}
\item[(i)] $ S \in  \mbox{Spin}(1,4), \qquad\gamma \in \Gamma_{+} \qquad \Rightarrow \qquad S \, \gamma \, S^\dag \in \Gamma_{+}  $
\item[(ii)] If $S$ belongs to the spin group and $V \in E$, we put
$$ \diamondsuit \qquad S \, \gamma (V) \, S^\dag = \gamma (P_S (V))
$$
then
$$ \heartsuit \qquad P_S \quad \mbox{is a rotation of} \quad \mathbb{O} (1,4)
$$
\item[(iii)] If $S$ belongs to the spin group, then there exist $r\in \mathbb{R}$ and unit quaternions $E, F, G$ such that
\begin{equation}
          S  = \left[ {{\begin{array}{cc}
       \sqrt{1 + r^2} \, 1_{\mathbb{C}^2} \hfill &  r \, \bar{E}    \hfill  \\
       r \, E \hfill &  \sqrt{1 + r^2} \, 1_{\mathbb{C}^2}   \hfill  \\
   \end{array} }} \right] \,
    \left[ {{\begin{array}{cc}
       F  \hfill &  0    \hfill  \\
       0  \hfill &  G    \hfill  \\
   \end{array} }} \right] \,
   = \left[ {{\begin{array}{cc}
       \sqrt{1 + r^2} \, F \hfill &  r \, \bar{E} \, G    \hfill  \\
       r \, E \, F \hfill &  \sqrt{1 + r^2}  \, G  \hfill  \\
   \end{array} }} \right]
\label{standard form of an element of Spin(1, 4)}
\end{equation}
\item[(iv)]  If $S$ belongs to the spin group, then there exist a quaternion $H$ and unit quaternions $F, G$ such that
\begin{equation}
          S  = \left[ {{\begin{array}{cc}
       \sqrt{1 + \mid H \mid^2} \, 1_{\mathbb{C}^2} \hfill &  \bar{H}    \hfill  \\
       H \hfill &  \sqrt{1 + \mid H \mid^2} \, 1_{\mathbb{C}^2}   \hfill  \\
   \end{array} }} \right] \,
    \left[ {{\begin{array}{cc}
       F  \hfill &  0    \hfill  \\
       0  \hfill &  G    \hfill  \\
   \end{array} }} \right] \,
   = \left[ {{\begin{array}{cc}
       \sqrt{1 + \mid H \mid^2} \, F \hfill &  \bar{H} \, G    \hfill  \\
       H  \, F \hfill &  \sqrt{1 + \mid H \mid^2} \, G  \hfill  \\
   \end{array} }} \right] 
\label{standard form 2 of an element of Spin(1, 4)}
\end{equation}
\end{itemize}
\end{theorem}
\textbf{Proof.}  
(i) If $S \in  \mbox{Spin}(1,4)$ and $\gamma \in \Gamma_{+}  $, $S \, \gamma \, S^\dag$ is a product of elements of $\Gamma$, then an element of $\Gamma$. We have $(S \, \gamma \, S^\dag)^\dag = (S^\dag)^\dag \, \gamma^\dag \, S^\dag = S \, \gamma \, S^\dag$ and $Tr ( S \, \gamma \, S^\dag) = Tr ( \gamma \, S^\dag S)  = Tr (\gamma) = 0 $ then $\gamma$ is an element of $\Gamma_{+} $.

(ii) The previous result shows that the relation $\diamondsuit$ defines for all $S$ of $\mbox{Spin}(1,4)$ a linear map $P_S$ from $E_{1,4}$ into $E_{1,4}$. Squaring $\diamondsuit$ and owing to (\ref{(gamma(X))^2 = (X^* X ) 1_(C^4)}), we obtain
$$ (V^* V)  \, 1_{\mathbb{C}^4} = (S_P (V))^* S_P (V))  \, 1_{\mathbb{C}^4} 
$$
that shows $P_S$ is a rotation of $\mathbb{O} (1,4)$.

(iii) If $S \in \mbox{Spin}(1,4)$, there exist quaternions $A, B, C, D$ such that
$$       S = \left[ {{\begin{array}{cc}
       A  \hfill &  C    \hfill  \\
       B  \hfill &  D    \hfill  \\
   \end{array} }} \right] 
$$
Developping the equation $S^\dag S = 1_{\mathbb{C}^4}$, we find
\begin{equation}
    \bar{A} \, A = 1_{\mathbb{C}^2} + \bar{B} \, B, \quad 
    \bar{D} \, D = 1_{\mathbb{C}^2} + \bar{C} \, C, \quad
    \bar{C} \, A = \bar{D} \, B 
\label{bar(A) A = 1 + bar(B) B & bar(D) D = 1 + bar(C) C & bar(C) A = bar(D) C}
\end{equation}
As for every quaternion $Q$ there exists $s \geq 0$ such that $\bar{Q} \, Q  = Q \, \bar{Q} = s \, 1_{\mathbb{C}^2}  $, we can state
$\bar{B} \, B = r^2  1_{\mathbb{C}^2},  \bar{C} \, C = r'^2  1_{\mathbb{C}^2} $. The former equation (\ref{bar(A) A = 1 + bar(B) B & bar(D) D = 1 + bar(C) C & bar(C) A = bar(D) C}) gives $ \bar{A} \, A = (1 + r^2) \, 1_{\mathbb{C}^2}$, allowing to put $ A = \sqrt{1 + r^2} \, F $, $F$ being a unit quaternion. Likewise, we find $ D = \sqrt{1 + r'^2} \, G $.

Inserting these results in the last equation (\ref{bar(A) A = 1 + bar(B) B & bar(D) D = 1 + bar(C) C & bar(C) A = bar(D) C}) and calculating the determinant of each term, we get $r^2 = r'^2 $. The equation is reduced to $\bar{C} \, F = \bar{G} \, B $ or  $C = F \, \bar{B} \, G$ but, as $\mid \det (B) \mid = r^2$, there exists a unit quaternion $E$ such that $B = r \, E \, F$ and consequently
$C = r \, \bar{E} \, G$, that proves (\ref{standard form of an element of Spin(1, 4)}).

Conversely, if $S$ is defined by (\ref{standard form of an element of Spin(1, 4)}), we find $S^\dag S = 1_{\mathbb{C}^4}$.

(iv) It is a consequence of (\ref{standard form of an element of Spin(1, 4)}) with the decomposition $H = r \, E$ then $\mid H \mid = \mid r \mid \, \mid E \mid = \mid r \mid$ that achieves the proof. $\blacksquare$


\vspace{0.5cm}

For the applications to Physics, we consider $E_{1, 3}  = \mathbb{R}^4$ equipped with the \textbf{Minkowski metric}  (of signature $1 + 3$),
\begin{equation}
     X^* X = (X^0)^2 - (X^1)^2 - (X^2)^2- (X^3)^2 = t^2 - (x_1)^2 - (x_2)^2- (x_3)^2
\label{Minkowski metric defi X^*X =}
\end{equation}
that gives it a structure of hyperbolic space. The \textbf{Lorentz group} $\mathbb{SO} (1,3) $ is the set of linear transformations $P$ of the 4D space conserving this metric
$$ P^* P = 1_{\mathbb{R}^4}
$$
called Lorentz transformations.

\begin{theorem}
With the above index convention, let us define the map
\begin{equation}
       \gamma: E_{1,3} \rightarrow \Gamma_{+}: X \mapsto \gamma (X) = X^i \gamma_i  
\label{gamma : E_(1,3) -> Gamma_(+) : X -> gamma(X) = X_i gamma_i}
\end{equation}
the matrices $\gamma_i$ being defined by (\ref{defi gamma_alpha}). 
\begin{itemize}
\item The map $\gamma$ is linear and injective.  It verifies (\ref{gamma(U) gamma(V) + gamma(V) gamma(U) = (U^* V) 1_(C^4)}) and
$$(\gamma (X))^2 =  (X^* X)  \, 1_{\mathbb{C}^4}  
$$  
\item For every Lorentz transformation $P$, there exists $S \in \mbox{Spin}(1,4)$ such that
$$ \gamma(P \, V) = S \, \gamma(V) \, S^{-1}
$$
$S$ is defined up to a factor $z$ root of unity by this identity and will be denoted $S(P)$
\item $ S (P)  \, \gamma_5 = \gamma_5 \, S(P)$
\item there exist $r\in \mathbb{R}$ and unit quaternions $E, F (\bar{E} = - E)$ such that
\begin{equation}
          S (P)  = \left[ {{\begin{array}{cc}
       \sqrt{1 + r^2} \, 1_{\mathbb{C}^2} \hfill &  - r \, E    \hfill  \\
       r \, E \hfill &  \sqrt{1 + r^2} \, 1_{\mathbb{C}^2}   \hfill  \\
   \end{array} }} \right] \,
    \left[ {{\begin{array}{cc}
       F  \hfill &  0    \hfill  \\
       0  \hfill &  F    \hfill  \\
   \end{array} }} \right] \,
   = \left[ {{\begin{array}{cc}
       \sqrt{1 + r^2} \, F \hfill &  - r \, E \, F    \hfill  \\
       r \, E \, F \hfill &  \sqrt{1 + r^2}  \, F  \hfill  \\
   \end{array} }} \right]
\label{standard form of S(P)}
\end{equation}
\item there exist a quaternion $H = - \bar{H}$ and a unit quaternions $F$ such that
\begin{equation}
          S  = \left[ {{\begin{array}{cc}
       \sqrt{1 + \mid H \mid^2} \, 1_{\mathbb{C}^2} \hfill &  - H   \hfill  \\
       H \hfill &  \sqrt{1 + \mid H \mid^2} \, 1_{\mathbb{C}^2}   \hfill  \\
   \end{array} }} \right] \,
    \left[ {{\begin{array}{cc}
       F  \hfill &  0    \hfill  \\
       0  \hfill &  F   \hfill  \\
   \end{array} }} \right] \,
   = \left[ {{\begin{array}{cc}
       \sqrt{1 + \mid H \mid^2} \, F \hfill &  - H \, F   \hfill  \\
       H \, F \hfill &  \sqrt{1 + \mid H \mid^2} \, F   \hfill  \\
   \end{array} }} \right] 
\label{standard form 2 of S(P)}
\end{equation}
\end{itemize}
\label{Thm S(P) for Lorents transformations}
\end{theorem}

In the last formula, the left matrix (containing $H$) generates a Lorentz boost, while the right matrix (containing $F$) generates a rotation of the 3D space. The reader can find the proof in \S 31.B of \cite{CL},

\section{Vector product}
\label{Section - Vector product}

We consider an oriented Euclidean space $\mathcal{T}$ of dimension $n$, {\it i.e.} there exists a volume form $vol \in \bigwedge^n \mathcal{T}^*$ such that $vol(\bm{e}_1, \ldots, \bm{e}_{n}) = 1$
for every orthonormal basis $(\bm{e}_i)$. We denote  $vol(\bm{e}_1, \ldots, \bm{e}_{q})$  the $(n - q)$-form such that 
$$ (vol(\bm{e}_1, \ldots, \bm{e}_{q})) (\bm{e}_{q + 1}, \ldots, \bm{e}_{n}) = vol (\bm{e}_1, \ldots, \bm{e}_{q},\bm{e}_{q + 1}, \ldots, \bm{e}_{n})
$$
We call \textbf{vector product} of $(n - 1)$ vectors $\bm{V}_1, \ldots, \bm{V}_{n - 1}$ the vector $\mathcal{J}(\bm{V}_1, \ldots, \bm{V}_{n - 1})$ such that 
$$ \mathcal{J}
    (\bm{V}_1, \ldots, \bm{V}_{n - 1})^* \bm{U} 
    = vol(\bm{V}_1, \ldots, \bm{V}_{n - 1},\bm{U}) 
$$
In \S 26.B of \cite{CL}, it is proved that

\begin{theorem}[Souriau]
Properties of the vector product:
\begin{itemize}
\item [$\diamondsuit$] $\mathcal{J}(\bm{V}_1, \ldots, \bm{V}_{n - 1}) \neq \bm{0}$ if and only if  $\bm{V}_1, \ldots, \bm{V}_{n - 1}$  are linearly independent
\item [$\heartsuit$] $\mathcal{J}(\bm{V}_1, \ldots, \bm{V}_{n - 1})$ is orthogonal to every argument $\bm{V}_i$ of $\mathcal{J}$
\item [$\spadesuit$] The linear map $\mathcal{J}(\bm{V}_1, \ldots, \bm{V}_{n - 2}) : \mathcal{T} \rightarrow \mathcal{T} : \bm{V} \mapsto \mathcal{J}(\bm{V}_1, \ldots, \bm{V}_{n - 2}, \bm{V})$ is one-to-one and skew-adjoint
\end{itemize}
\label{thm vector product}
\end{theorem}

In the particular case of the classical positive Euclidean space of dimension 3 ($n = p = 3$), we use the notation $j$ instead of $\mathcal{J}$ and we recover the cross product $j(u,v) = u \times v$ and the linear map $u \mapsto j(u)$ is one-to-one from $\mathbb{R}^3$ into the space of $3 \times 3$ skew-symmetric matrices. In the sequel, the notation $\mathcal{J}$ will be reserved to the dimension 4. In $E_{1, 3}$ equipped with the metric of Gram's matrix $G = diag(1, -1, -1, -1)$, the vector product of 3 vectors 
$$ \Pi_k = \left[ {{\begin{array}{cc}
       m_k \hfill  \\
       p_k  \hfill  \\
   \end{array} }} \right], \qquad m_k \in \mathbb{R}, \quad p_k \in \mathbb{R}^3 \qquad (1 \leq k \leq 3)
$$ 
itemizes as
\begin{equation}
   \mathcal{J}(\Pi_1, \Pi_2, \Pi_3) =  \left[ {{\begin{array}{cc}
       vol (p_1, p_2, p_3) \hfill  \\
        m_1 j (p_2, p_3) - m_2 j (p_1, p_3) + m_3 j (p_1, p_2) \hfill  \\
   \end{array} }} \right]   
\label{J (Pi_1, Pi_2, Pi_3) =}
\end{equation}
from which we deduce 
\begin{equation}
   \mathcal{J}(\Pi_1, \Pi_2) =  \left[ {{\begin{array}{cc}
       0            \hfill &  j(p_1, p_2)^T                   \hfill  \\
        j(p_1, p_2) \hfill &  m_1 j (p_2) - m_2 j (p_1) \hfill  \\
   \end{array} }} \right]   
\label{J (Pi_1, Pi_2) =}
\end{equation}

\begin{figure}[ht!]
\centering
\includegraphics[scale=.50]{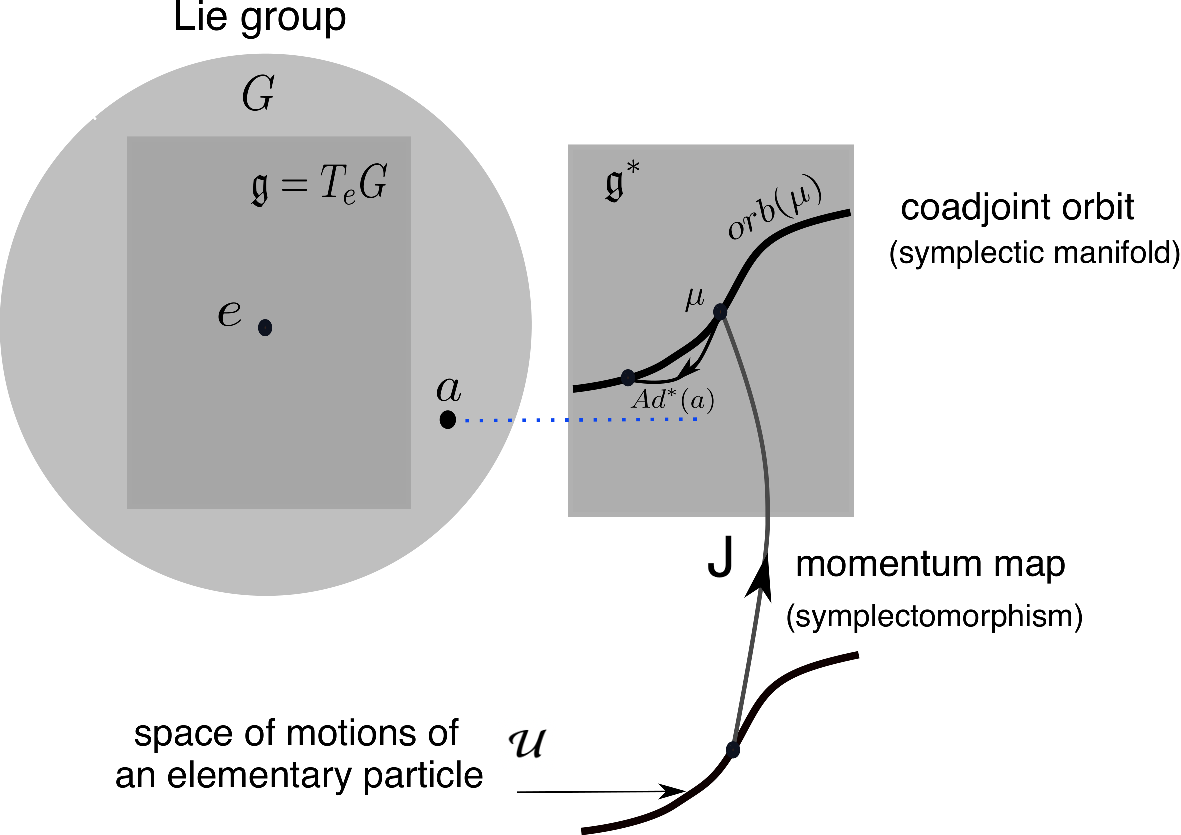}
\caption{Coadjoint orbit method}
\label{fig Coadjoint orbit method}
\end{figure}

\section{The coadjoint orbit method}
\label{Section - The coadjoint orbit method}

This method, the basis of the geometric quantization, was developed by Jean-Marie Souriau in his book "Structure of Dynamical Systems" \cite{SSD, SSDEng} to which the reader is referred for more details. The gist of the method is explained in Figure \ref{fig Coadjoint orbit method}. On the left, we have a symmetry group, the Lie group $G$ and its Lie algebra $\mathfrak{g}$, \textit{i.e.} the tangent space at the unity.
On the right, we have the dual $\mathfrak{g}^*$ of the Lie algebra on which act the elements of the group by the coadjoint representation. The orbit of a momentum $\mu$ has the structure of symplectic manifold.
The momentum map $\mathsf{J}$ is a symplectomorphism (\textit{i.e.} an isomorphism of the symplectic structure) from the space of motions $\mathcal{U}$ of an elementary particle into this orbit.
By classifying coadjoint orbits, one classifies the particles known to physicists, whether elementary such as quarks and leptons or composite such as hadrons.

\vspace*{0.5cm}

The \textbf{Poincaré group} $\mathbb{G} = \mathbb{R}^4 \rtimes \mathbb{SO} (1,3) $ is the set of affine transformations $a = (C, P)$ of the 4D space $E_{1, 3}$ conserving the Minkowski metric, composed of a translation $C$ and a linear part $P$ which is a Lorentz transformation
$$ P^* P = 1_{\mathbb{R}^4}
$$
$\mathbb{G}$ is a Lie group of dimension $10$. The elements of the Lie algebra have a skew-adjoint linear part $\delta P$
$$ \delta a = (\delta C, \delta P) \in  \mathfrak{g}  
 \qquad   \Leftrightarrow \qquad 
 \delta P^* = - \delta P
$$
The corresponding momenta $\mu$ are linear forms on $\mathfrak{g}$, of components the linear 4-momentum $\Pi$ and the angular 4-momentum $M$
\begin{equation}
    \mu (Z) = - \Pi^* \delta C  - \frac{1}{2} Tr (M \, \delta P)
\label{mu (Z) = - Pi delta C - (1/2) Tr (M delta P)}
\end{equation}
characterized by
$$ \mu \in \mathfrak{g}^*   
\qquad   \Leftrightarrow  \quad 
\mu = (\Pi, M) \quad \mbox{such that} \quad  M= - M^*
$$
The coadjoint representation reads
\begin{equation}
     \mu = Ad (a) \, \mu' \quad \Leftrightarrow \quad 
 \Pi = P \, \Pi', \quad 
 M = P \, M' P^* + C \, (P \, \Pi')^* - (P \, \Pi') \, C^*
\label{mu = Ad (a) mu'}
\end{equation}

\begin{figure}[ht!]
\centering
\includegraphics[scale=.60]{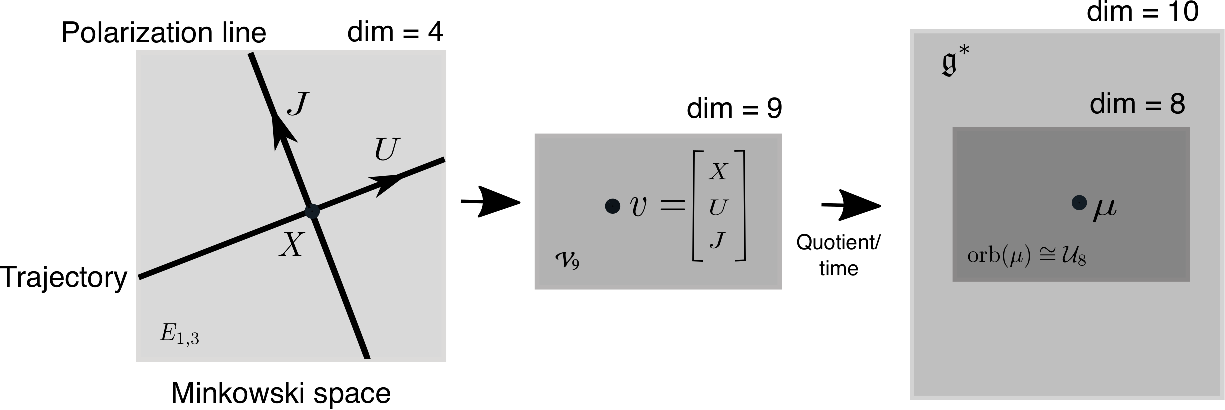}
\caption{Coadjoint orbit of a particle with spin for Poincar\'e group}
\label{fig Particle with spin Poincare}
\end{figure}

In Chapter 3, \S 14 of \cite{SSD, SSDEng}, Souriau uses the type (timelike or lightlike) of the 4-momentum vector (or energy-momentum vector) $\Pi$ and obtains a large part of physicists' classification of elementary particles. In the sequel, we are only interested by the case where $\Pi$ is {\bf timelike}. 

For the {\bf relativistic particle with spin}, we define the {\bf spin momentum}
\begin{equation}
    M_0 = M + \Pi \, X^* - X \, \Pi^*
\label{M_0 = M + Pi X^* - X Pi^* spin-momentum def}
\end{equation}
from which we deduce its transformation law
\begin{equation}
    M_0 = P \, M'_0 P^*
\label{transformation law of the spin M_0}
\end{equation}
$M_0$ is skew-adjoint and the set of $X \in \mathbb{R}^4$ such that $M_0 \, \Pi = 0$ is a straight line $\mathcal{D}$ parallel to $\Pi$ representing the particle trajectory (Figure \ref{fig Particle with spin Poincare}). Then $M_0$ does not depends on $X$ when $X$ runs over $\mathcal{D}$. Taking into account that a non null vector orthogonal to a timelike vector is spacelike (see \S 28 of \cite{CL}), we introduce $U, J \in E_{1, 3}$ such that 
\begin{equation}
    \Pi = m_0 \, U, \qquad U^* U = 1,\qquad J^* J = -1, \qquad J^* U = 0
\label{ Pi = m I & I^* I = 1 & J^* J = -1 & J^* I = 0}
\end{equation}
On this basis, we put
\begin{equation}
    M_0 = s \, \Omega \qquad \mbox{with} \quad \Omega = \mathcal{J} (U, J)
\label{M_0 = s J (I, J)}
\end{equation}
where the scalar $s$ is the \textbf{spin} and the vector product  $\Omega$ of $I$ and $J$ has the same transformation law as that of $M_0$
\begin{equation}
    \Omega = P \, \Omega' P^*
\label{transformation law of Omega}
\end{equation}
The momentum $\mu = (\Pi, M)$ is characterized by $\mathcal{D}$  and $J$ or, equivalently, by $X, U, J$. Using the Hodge operator $*$, we can introduce the \textbf{polarization}  (or spin vector) 
$$ W = (*M) \, \Pi
$$
such that $W = s \, m_0 \, J$. Using Kirillov-Kostant-Souriau theorem (\cite{SSD, SSDEng}, (11.34)), Souriau obtains the explicit expression of the symplectic form for the orbit of the relativistic particle of \textbf{rest mass} $m_0$ and \textbf{spin} $s$
\begin{equation}
     \omega_\mathcal{U} (\delta \mu,  d \mu) =
- s \, Tr (\delta \Omega \, \Omega \, d\Omega)
 + m_0 (\delta X^* d U - d X^* \delta U)
\label{Symplectic form for the relativistic particle of spin s}
\end{equation}

As represented in Figure \ref{fig Particle with spin Poincare}, every inertial motion of a particle with spin is defined by $v$, composed of the position in the space-time, $X$, called an event, the $4$-velocity $U$ (carried by the trajectory) and a $4$-vector of polarization $J$ (carried by a polarization line) for a total of $3 \times 4 = 12$ components.
$U$ and $J$ verifying 3 conditions of orthogonality and normalization, the set 
\begin{equation}
   \mathcal{V}_9 = \lbrace v = \left[ {{\begin{array}{cc}
       X \hfill  \\
       U  \hfill \\
       J  \hfill \\
   \end{array} }} \right] \;\; X, U, J \in E_{1, 3} \;\;
   \mbox{such that} \;\;   U^* U = 1,\;\;  J^* J = -1, \;\;  J^* U = 0
   \rbrace
\label{V_9 defi}
\end{equation}
is a manifold of dimension $12 - 3 = 9$. Quotienting by the time, the \textbf{space of motions} $\mathcal{U}_8$ is of dimension $8$ and, by diffeomorphism, the coadjoint orbit as well.
In the dual space of dimension $10$, that of the Poincaré's group, the orbit is defined by $10 - 8 = 2$ independent invariants which may be the rest mass and the spin
$$ m_0 = \sqrt{\Pi^* \Pi}, \qquad
     s = \frac{\sqrt{- W^* W}}{\sqrt{\Pi^* \Pi}}
$$
Then $m_0$ and $s$ are positive. There are other orbits of different dimensions characterizing the spinless particles and the massless particles but which will not be considered in this paper. 

\begin{figure}[ht!]
\centering
\includegraphics[scale=.70]{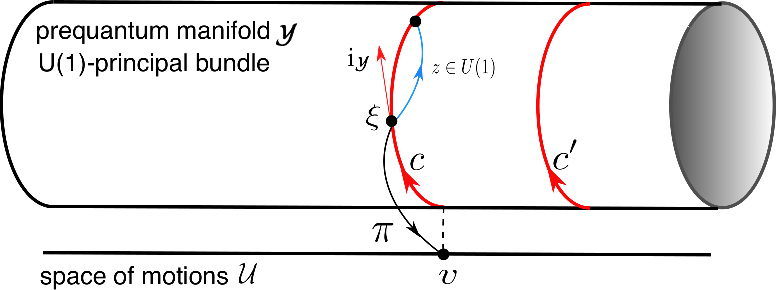}
\caption{Prequantum manifold}
\label{fig Prequantum manifold}
\end{figure}

\section{Geometric prequantization}
\label{Section - Geometric prequantization}

\begin{definition}
Let $\alpha$ a smooth field of 1-forms on a manifold $\mathcal{Y}$. If the following properties are satisfied
\begin{itemize}
\item [$\diamondsuit$] $\mbox{dim} (\mbox{ker} (\mbox{d}\alpha)) = 1$
\item [$\heartsuit$] $\mbox{dim} (\mbox{ker} (\alpha) \cap \mbox{ker} (\mbox{d}\alpha)) = 0$
\end{itemize}
the field  $\alpha$ is said to define a \textbf{contact structure} on $\mathcal{Y}$.
\label{defi contact structure}
\end{definition}
 If we define 
 \begin{equation}
      \omega_\mathcal{Y} = \mbox{d}\alpha
 \label{omega_Y = d varpi}
 \end{equation}
then $\omega_\mathcal{Y}$ is a 2-form whose rank is even. Condition \ref{defi contact structure} $\diamondsuit$ then implies that the dimension of $\mathcal{Y}$ is odd. As $\mbox{dim} (\mbox{ker} (\omega_\mathcal{Y}))$ is constant, $\mathcal{Y}$ is a presymplectic manifold. Its leaves are curves (lines of forces). We introduce the vector $\mbox{i}_\mathcal{Y}$ tangent to the  line of force (Figure \ref{fig Prequantum manifold})
\begin{equation}
     \alpha (\mbox{i}_\mathcal{Y}) = 1 \qquad \mbox{and} \qquad 
\lbrack \forall \delta \xi \in T_\xi \mathcal{Y} , \quad \omega_\mathcal{Y} (\mbox{i}_\mathcal{Y}, \delta \xi) = 0 \rbrack
\label{varpi (i_Y) = 1 and forall delta xi in T_xi Y  omega_Y (i_Y, delta xi) = 0}
\end{equation}
One should also note that
$$     z = e^{i \, s} \in \mathsf{U}(1) \quad \Rightarrow \quad z \cdot \xi = \mbox{exp} (s \, \mbox{i}_\mathcal{Y})
$$
Then every orbit $c$ of $\mathsf{U}(1)$ is a loop.  Because of (\ref{varpi (i_Y) = 1 and forall delta xi in T_xi Y  omega_Y (i_Y, delta xi) = 0}), we have
$$ \int_c \alpha (d \xi) \neq 0
$$
In the sequel, we choose the parameterization in such way that this quantity, called \textbf{action integral} is non negative. Then we can prove the following result

\begin{theorem}
Let $c$ and $c'$ be two orbits of $\mathsf{U}(1)$ on the prequantum manifold $\mathcal{Y}$. Then the action integral is conserved
$$ \forall c, c' \qquad \int_c \alpha (d \xi)  = \int_{c'} \alpha (d \xi) 
$$
\label{thm action conservation)}
\end{theorem}

\textbf{Proof.}  Let $\mathcal{S}$ be a submanifold of dimension 2 of which the boundary $\partial \mathcal{S}$ is the union of $c$ and $c'$ in such way that the parameterization of $\partial \mathcal{S}$ is the same as that of $c$ and opposite to that of $c'$. Applying Stokes Formula and taking into account (\ref{omega_Y = d varpi}), it  holds
$$ \int_c \alpha (d \xi)  -  \int_{c'} \alpha (d \xi) = \int_{\partial \mathcal{S}} \alpha (d \xi)  
   = \int_{\mathcal{S}} \omega_{\mathcal{Y}} (d \xi, \delta \xi)  
$$
To integrate, let us choose one of the two parameters of $\mathcal{S}$ being that of an orbit of $\mathsf{U}(1)$ and $d \xi$ tangent to it then proportional to $\mbox{i}_\mathcal{Y}$. Thus, because of (\ref{varpi (i_Y) = 1 and forall delta xi in T_xi Y  omega_Y (i_Y, delta xi) = 0}), we find
$$  \int_c \alpha (d \xi)  -  \int_{c'} \alpha (d \xi) =  0
$$
that achieves the proof. $\blacksquare$

\vspace{0.5 cm}

\begin{definition}
A manifold $\mathcal{Y}$ will be called a \textbf{prequantum manifold} if there exists a smooth field of 1-forms $\alpha$ on $\mathcal{Y}$ which defines a contact structure on $\mathcal{Y}$ and the torus $\mathsf{U}(1)$ acts on $\mathcal{Y}$ in such way that
\begin{itemize}
\item[$\bowtie$] $\mathcal{Y}$ is a a $\mathsf{U}(1)$-principal bundle of base $\mathcal{U}$ and projection $\pi$
\item [$\clubsuit$] $\forall \delta \xi, \quad \omega_\mathcal{Y} (\mbox{i}_\mathcal{Y}, \delta \xi) = 0$
\item [$\spadesuit$] $\alpha (\mbox{i}_\mathcal{Y}) = 1$
\item [$\heartsuit$]  $ \int_c \alpha (d \xi)   =  n \, h, \qquad n \in \mathbb{N}$
\end{itemize}
\label{defi prequantum manifold}
\end{definition}
Based on Theorem \ref{thm action conservation)},  the \textbf{Sommerfeld quantization condition} $\heartsuit$ postulates that the action integral is an integer multiple of the quantum of action of which the value is the Planck constant
$$ h = 2 \, \pi \, \hslash = 6.626 \, 10^{- 34} \,  J. s 
$$
As we are working in natural units, this condition becomes
\begin{equation}
     \int_c \alpha (d \xi)   = 2 \, \pi \, n, \qquad n \in \mathbb{N}
\label{Sommerfeld quantization condition BIS}
\end{equation}

As the lines of force  (which are the orbits of $\mathsf{U}(1)$) are compact, the characteristic foliation of $\mathcal{Y}$ is sectionable, then the set $\mathcal{U}$ of line of forces admits the structure of a symplectic manifold for $\omega_\mathcal{U}$ defined by
$$ \omega_\mathcal{U} (d x , \delta x) = \omega_\mathcal{Y}  (d\xi , \delta \xi) 
$$
where $x = \pi (\xi)$ is the projection of $\mathcal{Y}$ onto $\mathcal{U}$. 

\begin{definition} Conversely, let $\mathcal{U}$ a symplectic manifold representing the \textbf{space of motions of a particle}. We call \textbf{prequantization} of $\mathcal{U}$ the construction of a prequantum manifold $\mathcal{Y}$ and a symplectomorphism $A$ from the base of $\mathcal{Y}$ into $\mathcal{U}$. If we write $x = \mbox{orb} (\xi)$, then
\begin{itemize}
\item $\mathcal{Y}$ is a prequantum manifold
\item $\pi$ is a smooth map from $\mathcal{Y}$ into $\mathcal{U}$
\item $\mbox{ker} (T_\xi \pi)$ is generated by $\mbox{i}_\mathcal{Y}$
\item $x \in \mathcal{U} \quad \Rightarrow \quad \pi^{-1} (x)$ is an orbit of $\mathsf{U}(1)$
\item $\omega_\mathcal{Y}  (d\xi , \delta \xi) = \omega_\mathcal{U} (d x , \delta x) $
\end{itemize}
\label{defi prequantization}
\end{definition}
Thus, prequantizing a symplectic manifold $\mathcal{U}$ is the same as constructing a $\mathsf{U}(1)$-principal bundle $\pi : \mathcal{Y} \rightarrow \mathcal{U}$ satisfying the axioms of Definition \ref{defi prequantization}.

\section{Prequantization of the relativistic particle with  spin 1/2}
\label{Section - Prequantization of the relativistic particle with  spin 1/2}

Now, applying this method, we prequantize the space of motions $\mathcal{U}$ of the \textbf{relativistic particle with spin} described in Section \ref{Section - The coadjoint orbit method}. Especially, we would like to give more details concerning (\cite{SSD, SSDEng}, (18.67) to (18.70). In  the space $\mathcal{E}$ of spinors, the idea is to identify a subset $\Sigma_6$ thereof that can be projected onto the space of couples $(U, J)$ (see Figure \ref{fig Prequantization of the particle with spin with spin 1/2}). To guide us, we use the canonical decomposition (\ref{Gamma = Gamma-i oplus Gamma_(+) oplus Gamma_(-)}) of the space $\Gamma$ of quaternionic matrices.

\begin{figure}[ht!]
\centering
\includegraphics[scale=.60]{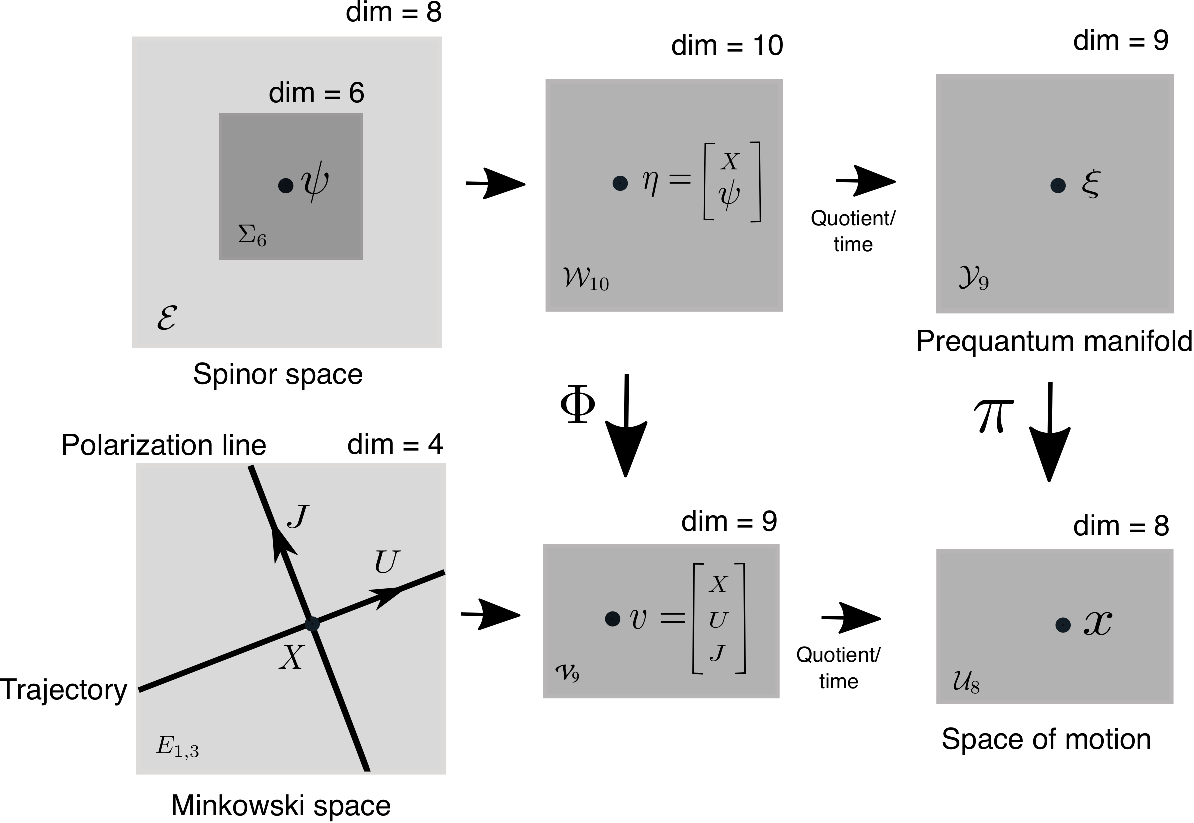}
\caption{Prequantization of the particle with spin with spin 1/2}
\label{fig Prequantization of the particle with spin with spin 1/2}
\end{figure}

\textbf{Step 1.} Interpreting the spinor as a state vector of the quantum mechanics, we would like to normalize it having a glance at the subspace  $\Gamma_i$ of isotropic matrices, spanned by the basis matrix $1_{\mathbb{C}^4}$. Because it is Hermitian, the scalar $\psi^\dag 1_{\mathbb{C}^4} \psi = \psi^\dag \psi $ is real, then we impose the \textbf{normalization condition} 
\begin{equation}
     \psi^\dag \psi = \mid \psi' \mid^2 - \mid \psi'' \mid^2 =  \epsilon, \qquad \epsilon = \pm 1
\label{normalization condition}
\end{equation}
Morover, according to (\cite{SSD, SSDEng}, (18.65)), we impose the condition
\begin{equation}
     \psi^\dag \gamma_5 \, \psi  = - 2 \, \Re (\psi'^\dag \psi'' )  = 0
\label{psi^dag gamma_5 psi = 0}
\end{equation}
Then the manifold of dimension 6
$$\Sigma_6 = \left\lbrace \psi \in \mathbb{C} ^4\quad \mbox{such that} \quad \psi^\dag \psi  = \epsilon 
     \quad \mbox{and} \quad \psi^\dag \gamma_5 \, \psi = 0 \right\rbrace
$$ 
is decomposed into two disjoint connected components
$$ \Sigma_{+}= \left\lbrace \psi \in \Sigma_6 \;\; \mbox{such that} \;\;  \psi^\dag \psi  = 1 \right\rbrace, \qquad
     \Sigma_{-} = \left\lbrace \psi \in \Sigma_6 \;\; \mbox{such that} \;\;  \psi^\dag \psi  = -1 \right\rbrace
$$
Unlike Souriau in (\cite{SSD, SSDEng}, (18.65)), we do not exclude {\it a priori} $\Sigma_{-}$.

\textbf{Step 2.}  Because the matrices $\gamma$ of the subspace $\Gamma_{+}$ are Hermitian, the scalar $\psi^\dag \gamma \, \psi$ is real. As the values of the map (\ref{gamma from R^5 into Gamma_(+)}) belong to $\Gamma_{+}$, according to Souriau in (\cite{SSD, SSDEng}, (18.68))$_1$), we can associate to every spinor $\psi$ a 4-vector $I$ such that
\begin{equation}
    \forall dX \in \mathbb{R}^4, \qquad I^* dX = \psi^\dag \gamma (dX) \, \psi 
\label{I^* dX = psi^dag gamma (dX) psi}
\end{equation}
Alternatively,  introducing the Hermitian matrices
\begin{equation}
      \gamma^j = G^{jm} \gamma _m
\label{gamma^j = G^(jm) gamma_m}
\end{equation}
where $G^{ij} $ are the components of the matrix $G^{-1}$, the components of $I$ are
\begin{equation}
       I^k = \psi^\dag \gamma^k \psi
\label{I^j = psi^dag gamma^j psi}
\end{equation}
As discussed in Section \ref{Section - Quantization of the relativistic particle with  spin 1/2}, we shall recover with this formula the classical concept of probability current (\cite{Bjorken 1964}., p. 9). In particular, for the Minkowski metric (\ref{Minkowski metric defi X^*X =}), it holds
$$ I^0 = \psi^\dag \gamma_0 \, \psi, \qquad
     I^k =  - \psi^\dag \gamma_k \, \psi\quad (k = 1, 2, 3) 
$$
Let $U$ be a \textbf{futur-directed 4-vector} such that 
\begin{equation}
    U^* U = 1
\label{U^* U = 1}
\end{equation}
and $\psi$ be an \textbf{eigenvector} of $\gamma(U)$. As this matrix is Hermitian, the corresponding \textbf{eigenvalue} $\epsilon_U$ is real. Taking into account (\ref{(gamma(X))^2 = (X^* X ) 1_(C^4)}), we have 
\begin{equation}
    \epsilon_U = \pm \, U^* U = \pm 1
\label{epsilon_U = pm U^* U = pm 1}
\end{equation}
Owing to (\ref{psi^dag gamma(U) psi > 0 if U is future-directed}) and (\ref{normalization condition}), it holds
$$ \psi^\dag \gamma (U) \, \psi  = \epsilon_U \, \psi^\dag \psi = \epsilon_U \epsilon > 0
$$
Consequently, $\epsilon_U$ and $\epsilon$ have the same sign, then, owing to (\ref{normalization condition}) and (\ref{epsilon_U = pm U^* U = pm 1})
$$\epsilon_U = \epsilon
$$
and
\begin{equation}
    \gamma (U) \, \psi = \epsilon \, U
\label{gamma (U)  psi = epsilon psi}
\end{equation}

\textbf{Step 3.}  Because the matrices $\gamma$ of the subspace $\Gamma_{-}$ are anti-Hermitian, the scalar $\psi^\dag \gamma \, \psi$ is imaginary. As the matrices $\gamma_\alpha \gamma_5$ belong to $\Gamma_{-}$ and taking into account (\ref{gamma from R^5 into Gamma_(+)}), according to Souriau in (\cite{SSD, SSDEng}, (18.68))$_2$), we can associate to every spinor $\psi$ a 4-vector $J$ such that
\begin{equation}
        \forall dX \in \mathbb{R}^4, \qquad J^* dX = i \,\psi^\dag \gamma (dX) \, \gamma_5 \,  \psi 
\label{J^* dX = psi^dag gamma (dX) gamma_5 psi}
\end{equation}
Alternatively,  the components of $J$ are given by
$$       J^k = i \,\psi^\dag \gamma^k \gamma_5 \, \psi
$$
a formula which, although similar to (\ref{I^j = psi^dag gamma^j psi}), is not generally mentioned in classical textbook but will give rise also to a conservation identity as shown in Section \ref{Section - Quantization of the relativistic particle with  spin 1/2}.

\vspace{0.5cm}

Before to go any further, owing to the decomposition (\ref{psi = (psi' psi'')})  of $\psi$ into half-spinors
$$  \psi = \left[ {{\begin{array}{c}
           \psi'                           \hfill  \\
           \psi''                           \hfill  \\
   \end{array} }} \right] 
$$
it is useful to verify that Formulae (\ref{I^* dX = psi^dag gamma (dX) psi}) and (\ref{J^* dX = psi^dag gamma (dX) gamma_5 psi})
give respectively the expressions of the components of $I$ and $J$
\begin{equation}
   I^0 = \mid \psi' \mid^2 + \mid \psi'' \mid^2, \qquad
   I^k = 2 \Im (\psi''^\dag \sigma_k \psi') \quad (k = 1, 2, 3) 
\label{I^0 = & I^k =}
\end{equation}
\begin{equation}
   J^0 = 2 \Im (\psi''^\dag \psi'), \qquad
   J^k = - (\psi'^\dag \sigma_k \psi' + \psi''^\dag \sigma_k \psi''),  \quad (k = 1, 2, 3) 
\label{J^0 = & J^k =}
\end{equation}

\vspace{0.5cm}

On this basis, we can prove the following result

\begin{theorem}
If the $U$ is a \textbf{futur-directed 4-vector} such that $U^* U = 1$,
there is equivalence between the two propositions:
\begin{itemize}
\item [$\diamondsuit$]  There exists a spinor $\psi \in \Sigma_6$, \textbf{eigenvector} of $\gamma(U)$ such that the 4-vectors $I$ and $J$ are respectively defined by (\ref{I^* dX = psi^dag gamma (dX) psi}) and (\ref{J^* dX = psi^dag gamma (dX) gamma_5 psi})
\begin{equation}
       \forall dX \in \mathbb{R}^4, \qquad I^* dX = \psi^\dag \gamma (dX) \, \psi 
\label{I^* dX = psi^dag gamma (dX) psi BIS}
\end{equation}
\begin{equation}
       \forall dX \in \mathbb{R}^4, \qquad J^* dX = i \,\psi^\dag \gamma (dX) \gamma_5 \psi 
\label{J^* dX = psi^dag gamma (dX) gamma_5 psi BIS}
\end{equation}
\item [$\heartsuit$] The 4-vectors $I$ and $J$ verify 
\begin{equation}
     I = U,  \qquad U^* U = 1,\qquad J^* J = -1, \qquad J^* U = 0
\label{I = U & U^* U = 1 & J^*J = -1 & J^*U = 0}
\end{equation}
\end{itemize}
\label{thm link psi <-> (I  J)}
\end{theorem}

\textbf{Proof.}  
\textbf{(i) Let us demonstrate that} $\; \diamondsuit \; \Rightarrow \;\heartsuit$. Indeed, $U$ is of the form
$$ U = \left[ {{\begin{array}{c}
           \gamma _u                       \hfill  \\
           \gamma_u u                       \hfill  \\
   \end{array} }} \right], \qquad \mbox{with} \quad u \in \mathbb{R}^3, \quad \gamma_u = (1 - \parallel u \parallel^2)^{-1/2}
$$
then, applying a Lorentz transformation of boost  $- u$  to $U$, we obtain  $\tilde{U} = P^{-1} U$ such that  
\begin{equation}
    U  = P\, \tilde{U}  =  \left[ {{\begin{array}{cc}
        \gamma _u    \hfill &  \gamma_u u^T    \hfill  \\
       \gamma_u u  \hfill &    1_{\mathbb{R}^4} + \frac{(\gamma_u)^2}{\gamma_u + 1}  \, u \, u^T\hfill  \\
   \end{array} }} \right] 
\left[ {{\begin{array}{c}
           1                       \hfill  \\
           0                       \hfill  \\
   \end{array} }} \right]
\label{U = P tilde(U) with P is Lorentz of boost u}
\end{equation}
By the choice of a suitable orthonormal basis, of $E_{1,3}$, as the column $\psi$ represents a spinor of the Hermitian space $\mathcal{E}$ and the matrix  $\bm{\gamma} (U)$  represents an element  of $End(\mathcal{E})$, according to Theorem \ref{Thm S(P) for Lorents transformations} with $S = S(P)$ given by (\ref{standard form of S(P)}) where
$$ r = \sqrt{\frac{1}{2} \, (\gamma_u - 1)}, \qquad 
     E = \frac{1}{\parallel u \parallel} \, (u_1 I + u_2 J + u_3 K), \qquad 
     F = 1_{\mathbb{C}^2}
$$
one has 
\begin{equation}
     \tilde{\psi}  = S^{-1} \psi = S^\dag \psi, \qquad
     \gamma (\tilde{U} ) = S^{-1} \gamma (U) \, S = S^\dag \gamma (U) S
\label{tilde(psi) = S^(-1) psi = S^dag psi & gamma(tilde(U)) =}
\end{equation}
then
$$ \forall d\tilde{X}  = P^{-1} dX \in \mathbb{R}^4, \qquad 
     \tilde{\psi} ^\dag \gamma (d\tilde{X} ) \, \tilde{\psi} 
      = (S^\dag \psi)^\dag (S^\dag \gamma (P^{-1} dX) \, S) \, (S^\dag \psi)
       = \psi^\dag \gamma (dX) \, \psi 
$$
Consequently, $\tilde{I}  = P^{-1} I$ is such that (\ref{I^* dX = psi^dag gamma (dX) psi BIS}) is equivalent to
\begin{equation}
    \forall d\tilde{X}   \in \mathbb{R}^4, \qquad \tilde{I} ^* d\tilde{X}  = \tilde{\psi} ^\dag \gamma (d\tilde{X} ) \, \tilde{\psi} 
\label{I^* dX = psi^dag gamma (dX) psi TER}
\end{equation}
Likewise, we can prove that $\tilde{J}  = P^{-1} J$ is such that (\ref{J^* dX = psi^dag gamma (dX) gamma_5 psi BIS}) is equivalent to 
\begin{equation}
    \forall d\tilde{X}  \in \mathbb{R}^4, \qquad \tilde{J} ^* d\tilde{X}  
           = \tilde{\psi} ^\dag \gamma (d\tilde{X} ) \, \gamma_5 \, \tilde{\psi} 
\label{J^* dX = psi^dag gamma (dX) gamma_5 psi TER}
\end{equation} 
and, owing to (\ref{I^0 = & I^k =}), that $ \gamma (\tilde{U} ) = \gamma_0$ and $\tilde{\psi} $ is an eigenvector of $\gamma_0$ of same eigenvalue $\epsilon$ as $\psi$
$$ \gamma_0 \, \tilde{\psi}  = \epsilon \, \tilde{\psi} 
$$
One verifies that the eigenspace of $\epsilon = 1$ is of dimension 2 and its elements have a null lower half-spinor then, owing to (\ref{normalization condition})
\begin{equation}
  \gamma_0 \, \tilde{\psi}  = \tilde{\psi}  \;\; \mbox{and} \;\; \tilde{\psi} \in \Sigma_{+} \quad \Leftrightarrow \quad 
\tilde{\psi}  = \left[ {{\begin{array}{c}
           \psi'                  \hfill  \\
           0                       \hfill  \\
   \end{array} }} \right] \;\; 
   \mbox{with} \;\;  \psi' = \left[ {{\begin{array}{c}
           z'_1                  \hfill  \\
           z'_2                  \hfill  \\
   \end{array} }} \right] \;\;  \mbox{and} \;\;  \mid \psi' \mid^2 = 1
\label{gamma_0 tilde (psi) = tilde(psi) & tilde(psi) in Sigma_(+) equivalent to}
\end{equation}
and, using (\ref{I^0 = & I^k =}) and (\ref{J^0 = & J^k =}), we obtain
\begin{equation}
    \tilde{I}  = \left[ {{\begin{array}{c}
           1  \hfill  \\
           0  \hfill  \\
           0  \hfill  \\           
           0  \hfill  \\
   \end{array} }} \right], \qquad
    \tilde{J}  = \left[ {{\begin{array}{c}
           0  \hfill  \\
           - 2 \, \Re (z'_1 \overline{z'_2})               \hfill  \\
              2 \, \Im (z'_1 \overline{z'_2})               \hfill  \\           
           \mid z'_2 \mid^2  -  \mid z'_1 \mid^2  \hfill  \\
   \end{array} }} \right] 
\label{I = (1 0) & J = = (0 n) }
\end{equation}
It is worth to remark that 
$$ \tilde{J} ^* \tilde{J}  = - \lbrack 4 \, \mid z'_1 \overline{z'_2} \mid^2 + (\mid z'_2 \mid^2  -  \mid z'_1 \mid^2  )^2 \rbrack
$$
$$ \tilde{J} ^* \tilde{J}    = - \lbrack 4 \, \mid z'_1 \mid^2 \mid z'_2 \mid^2 
                             +( \mid z'_1 \mid^4  +  \mid z'_2 \mid^4  - 2 \, \mid z'_1 \mid^2 \mid z'_2 \mid^2 ) \rbrack
$$
and, owing to (\ref{normalization condition}), it holds for $\tilde{\psi}  \in \Sigma_{+}$
$$ \tilde{J} ^* \tilde{J}   = - (\mid z'_1 \mid^2  +  \mid z'_2 \mid^2  )^2) = - \mid \psi' \mid^4 = -1
$$
hence, for $\tilde{\psi}  \in \Sigma_{+}$, we have
\begin{equation}
     \tilde{I}  = \tilde{U} ,  \qquad \tilde{U} ^* \tilde{U}  = 1,\qquad \tilde{J} ^* \tilde{J}  = -1, \qquad \tilde{J} ^* \tilde{U}  = 0
\label{I' = U' & U'^* U' = 1 &  J'^*J' = -1 & J'^*U' = 0}
\end{equation}
and applying the boost $P$ to $\tilde{U} $, $\tilde{I} $ and $\tilde{J} $, we prove  the equivalence of (\ref{I' = U' & U'^* U' = 1 &  J'^*J' = -1 & J'^*U' = 0}) and  (\ref{I = U & U^* U = 1 & J^*J = -1 & J^*U = 0})
$$      I = U,  \qquad U^* U = 1,\qquad J^* J = -1, \qquad J^* U = 0
$$
By a similar reasoning, we prove that it is so also for the the eigenvectors of $\epsilon = -1$ of which elements have a null upper half-spinor
$$  \gamma_0 \, \tilde{\psi}  = - \tilde{\psi}  \;\; \mbox{and} \;\; \tilde{\psi} \in \Sigma_{-} \quad \Leftrightarrow \quad 
\tilde{\psi}  = \left[ {{\begin{array}{c}
           0                            \hfill  \\
           \psi''                      \hfill  \\
   \end{array} }} \right] \;\;
   \mbox{with} \;\; \psi'' = \left[ {{\begin{array}{c}
           z''_1                  \hfill  \\
           z''_2                  \hfill  \\
   \end{array} }} \right] \;\;  \mbox{and} \;\;  \mid \psi'' \mid^2 = 1
$$

\vspace{0.5cm}

\textbf{(ii) Conversely} $\; \heartsuit \; \Rightarrow \; \diamondsuit$. The 4-vectors $U, I, J$ being given, we consider the Lorentz transformation defined by (\ref{U = P tilde(U) with P is Lorentz of boost u}) to determine $\tilde{U} = P^{-1} U, \tilde{I} = P^{-1} I, \tilde{J} = P^{-1} J$. Then, taking into account (\ref{I = (1 0) & J = = (0 n) }), $\tilde{I}, \tilde{U}, \tilde{J}$ are of the form
$$    \tilde{I}  = \tilde{U} =  \left[ {{\begin{array}{c}
           1  \hfill  \\
           0  \hfill  \\
   \end{array} }} \right], \qquad
    \tilde{J}  = \left[ {{\begin{array}{c}
           0  \hfill  \\
           n  \hfill  \\
   \end{array} }} \right] \;\; \mbox{with} \;\; n \in \mathbb{R}^3, \;\; \parallel n \parallel = 1
$$
Let us consider the Lorentz transformation
$$ Q =  \left[ {{\begin{array}{cc}
        1    \hfill &  0    \hfill  \\
        0    \hfill &  R    \hfill  \\
   \end{array} }} \right] 
$$
where $R$ is the rotation of axis unit vector $\nu$ and angle $\theta$ which apply the vector $( -e_3)$ of the canonical basis of $\mathbb{R}^3$ onto $n$. Therefore we have
\begin{equation}
      \hat{I} = Q^{-1} \tilde{I} = \tilde{U} =  \left[ {{\begin{array}{c}
           1  \hfill  \\
           0  \hfill  \\
   \end{array} }} \right], \qquad
    \hat{J} = Q^{-1} \tilde{J}  = \left[ {{\begin{array}{c}
           0  \hfill  \\
        - e_3  \hfill  \\
   \end{array} }} \right] 
\label{hat(I) = (1 0) & hat(J) = (0 e_3)}
\end{equation}
According to Theorem \ref{Thm S(P) for Lorents transformations} with $S = S(Q)$ given by (\ref{standard form of S(P)}) where
$$ r = 0, \qquad 
     F = \cos (\theta / 2) \, 1_{\mathbb{C}^2} + \sin (\theta / 2) \, (\nu_1 I + \nu_2 J + \nu_3 K)
$$
taking into account (\ref{I^0 = & I^k =}) and (\ref{J^0 = & J^k =}), we see that $\hat{I}$ and $\hat{J}$ can be obtained by (\ref{I^* dX = psi^dag gamma (dX) psi BIS}) and (\ref{J^* dX = psi^dag gamma (dX) gamma_5 psi BIS}) from
\begin{equation}
 \hat{\psi} = \left[ {{\begin{array}{c}
           1  \hfill  \\
           0  \hfill  \\
           0  \hfill  \\           
           0  \hfill  \\
   \end{array} }} \right] \in \Sigma_{+}
\label{hat(psi) = (1 0 0 0) in Sigma_(+)}
\end{equation}
Hence $I$ and $J$ are respectively provided by   (\ref{I^* dX = psi^dag gamma (dX) psi}) and (\ref{J^* dX = psi^dag gamma (dX) gamma_5 psi}) from
$$ \psi = S(P) \, \tilde{\psi} =  S(P) \, S(Q) \, \hat{\psi}
$$
that achieves the proof. $\blacksquare$

\vspace{0.5cm}

\textbf{Remark 1.} In the proof of (ii), the method is constructive. However it is worth to observe that the spinor $\psi$ is not unique or otherwise the map $\psi \mapsto (I, J)$ defined by (\ref{I^* dX = psi^dag gamma (dX) psi BIS}) and (\ref{J^* dX = psi^dag gamma (dX) gamma_5 psi BIS})  is surjective. If we multiply $\hat{\psi}$ by $z \in \mathsf{U} (1)$, we find once again (\ref{hat(I) = (1 0) & hat(J) = (0 e_3)}). Also, we can obtain (\ref{hat(I) = (1 0) & hat(J) = (0 e_3)}) from
\begin{equation}
    \hat{\psi} = \left[ {{\begin{array}{c}
           0  \hfill  \\
           0  \hfill  \\
           1  \hfill  \\           
           0  \hfill  \\
   \end{array} }} \right] \in \Sigma_{-}
\label{hat(psi) = (0 0 0 1) in Sigma_(-)}
\end{equation}
or any one of its multiples by $z \in \mathsf{U} (1)$. 

\textbf{Remark 2.} It is worth to remark that the key of the demonstration is based on the \textbf{relativistic covariance} of the Formula (\ref{I^* dX = psi^dag gamma (dX) psi BIS}) (resp. (\ref{J^* dX = psi^dag gamma (dX) gamma_5 psi BIS})) resulting from the equivalence with (\ref{I^* dX = psi^dag gamma (dX) psi TER}) (resp. (\ref{J^* dX = psi^dag gamma (dX) gamma_5 psi TER})).

\vspace{0.5cm}

\begin{figure}[ht!]
\centering
\includegraphics[scale=.60]{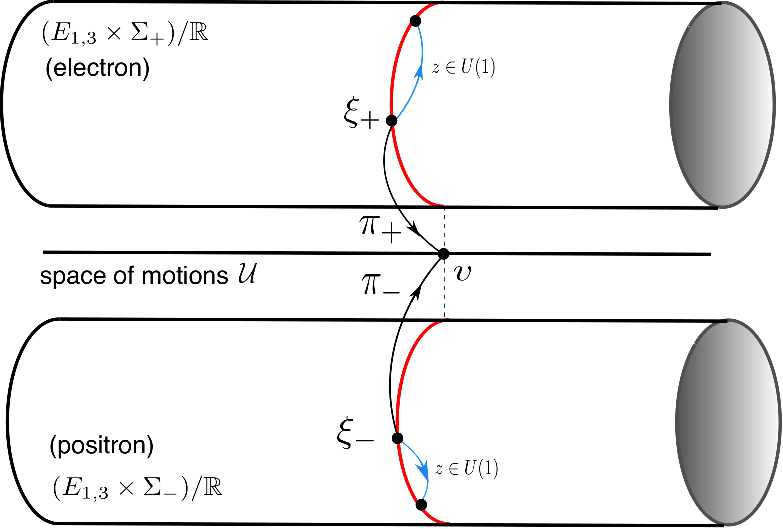}
\caption{Prequantum manifolds of the electron and the positron}
\label{fig Prequantum manifolds of the electron and the positron}
\end{figure}

The next step is to construct the manifold of dimension 10
\begin{equation}
        \mathcal{W}_{10} 
        = E_{1, 3} \times \Sigma_6
        = \lbrace    \;\; \eta = \left[ {{\begin{array}{cc}
       X      \hfill  \\
       \psi  \hfill \\
   \end{array} }} \right] \quad X \in E_{1, 3}, \qquad \psi \in \Sigma
   \;\;  \rbrace
\label{W_(10) =}
\end{equation}
where $\Sigma$ is either $\Sigma_{+}$ or $\Sigma_{-}$ .  Theorem \ref{thm link psi <-> (I  J)} shows there exists a surjective map from  $ \mathcal{W}_{10}$ (which is locally diffeomorphic to $\mathcal{V}_9 \times \mathsf{U}(1)$) to $\mathcal{V}_9$ defined by (\ref{V_9 defi}). It is worth to notice that, in this definition, we do not consider the union $\Sigma_6$ of $\Sigma_{+}$ and $\Sigma_{-}$ which does not allow to construct a $\mathsf{U}(1)$-principal bundle $\mathcal{Y}$ because the action of $\mathsf{U}(1)$ would not be transitive. Hence there are two distinct particles corresponding to the classical one (Figure \ref{fig Prequantum manifolds of the electron and the positron}). As proposed Dirac, they are the electron and the positron. In contrast, Souriau excludes $\Sigma_{-}$ in (\cite{SSD, SSDEng}, (18.65)).

\vspace{0.5cm}

In the left column of Table 2, we gathered eigenvectors of $\gamma_0$ belonging to $\Sigma_6$ which make up a $\mathbb{C}$-basis of the space of spinors and of which we already met two of that in (\ref{hat(psi) = (1 0 0 0) in Sigma_(+)}) and (\ref{hat(psi) = (0 0 0 1) in Sigma_(-)}). The two former belong to $\Sigma_{+}$ then represents a quantum state of the electron while the two later belong to $\Sigma_{-}$ then represents a quantum state of the positron. In the center column, we gives $U = I$ which is identical for the four states and is the 4-velocity of particles at rest in the frame of reference. The right column provides the corresponding 4-vector giving the spin orientation in the direction $x_3$: $-1$ for the first and third states, $+1$ for the second and fourth states.

\begin{center}
\begin{tabular}{l l l l l}
 \hline \hline
  $\psi$ & $U = I \in \mathbb{R}^4$ & $J \in \mathbb{R}^4$ & particle & spin \\
  \hline \hline
   & & & & \\
  $    \left[ {{\begin{array}{c}
           1  \hfill  \\
           0  \hfill  \\
           0  \hfill  \\           
           0  \hfill  \\
   \end{array} }} \right] \in \Sigma_{+}$  &  $\left[ {{\begin{array}{c}
           1 \hfill  \\
           0  \hfill  \\
           0  \hfill  \\           
           0  \hfill  \\
   \end{array} }} \right] $ &  $\left[ {{\begin{array}{c}
           0 \hfill  \\
           0  \hfill  \\
           0  \hfill  \\           
          -1  \hfill  \\
   \end{array} }} \right] $   & electron &    $\downarrow$      \\
    & & & & \\
     $    \left[ {{\begin{array}{c}
           0  \hfill  \\
           1  \hfill  \\
           0  \hfill  \\           
           0  \hfill  \\
   \end{array} }} \right] \in \Sigma_{+}$  &  $\left[ {{\begin{array}{c}
           1 \hfill  \\
           0  \hfill  \\
           0  \hfill  \\           
           0  \hfill  \\
   \end{array} }} \right] $ &  $\left[ {{\begin{array}{c}
           0 \hfill  \\
           0  \hfill  \\
           0  \hfill  \\           
          +1  \hfill  \\
   \end{array} }} \right] $    &  electron &   $\uparrow$      \\
   & & & & \\
     $    \left[ {{\begin{array}{c}
           0  \hfill  \\
           0  \hfill  \\
           1  \hfill  \\           
           0  \hfill  \\
   \end{array} }} \right] \in \Sigma_{-}$  &  $\left[ {{\begin{array}{c}
           1 \hfill  \\
           0  \hfill  \\
           0  \hfill  \\           
           0  \hfill  \\
   \end{array} }} \right] $ &  $\left[ {{\begin{array}{c}
           0 \hfill  \\
           0  \hfill  \\
           0  \hfill  \\           
          -1  \hfill  \\
   \end{array} }} \right] $    & positron &   $\downarrow$      \\
    & & & & \\
     $    \left[ {{\begin{array}{c}
           0  \hfill  \\
           0  \hfill  \\
           0  \hfill  \\           
           1  \hfill  \\
   \end{array} }} \right] \in \Sigma_{-}$  &  $\left[ {{\begin{array}{c}
           1 \hfill  \\
           0  \hfill  \\
           0  \hfill  \\           
           0  \hfill  \\
   \end{array} }} \right] $ &  $\left[ {{\begin{array}{c}
           0 \hfill  \\
           0  \hfill  \\
           0  \hfill  \\           
         +1  \hfill  \\
   \end{array} }} \right] $    & positron  &  $\uparrow$ \\
 & & & & \\
    \hline \hline
\end{tabular}
\end{center}
\begin{center}
\textbf{Table 2:} eigenvectors of $\gamma_0$
\end{center}

To obtain the quantum states with a spin $+ 1$ or $-1$ in the direction of the other axis of the frame of reference, we can use an element  of the spin group generating a spatial rotation (see Theorem \ref{Thm S(P) for Lorents transformations}). 

For instance, the spin $+1$ of the electron in the directions $x_1$ and $x_2$ correspond respectively to the quantum states
$$  \frac{1}{\sqrt{2}} \, \left[ {{\begin{array}{c}
           1  \hfill  \\
        - 1  \hfill  \\
           0  \hfill  \\           
           0  \hfill  \\
   \end{array} }} \right], \; 
 \frac{1}{\sqrt{2}} \,  \left[ {{\begin{array}{c}
           1  \hfill  \\
         - i  \hfill  \\
          0  \hfill  \\           
          0  \hfill  \\
   \end{array} }} \right] \in \Sigma_{+}
$$

\section{Symplectic and contact structures of the prequantum manifold}
\label{Section - Symplectic and contact structures of the prequantum manifold}

First of all, taking into account (\ref{Im (Z^dag Z') - - Im(Z'^dag Z)}), let us remark that 
$$ \omega_\Sigma (d \psi, \delta \psi) = - \Im (d \psi^\dag \delta \psi) 
$$
is a 2-form. Considering (\cite{SSD, SSDEng}, (18.71)), we start with the symplectic structure on $\Sigma_6$ based on the following proposition

\begin{theorem}
If the spinor $\psi$ runs through $\Sigma_6$, we have the identity
\begin{equation}
   - \Im (d \psi^\dag \delta \psi) = \frac{1}{4} \, Tr (d \Omega \, \Omega \, \delta \Omega)
\label{Im((d psi)^dag delta psi) = (1/4) Tr (d Omega Omega delta Omega)}
\end{equation}
with, as in (\ref{M_0 = s J (I, J)}),
\begin{equation}
    \Omega = \mathcal{J} (U, J)
\label{Omega = Omega = J(U, J)}
\end{equation}
\label{thm symplectic structure of Sigma_6}
\end{theorem}

\textbf{Proof.}  The idea is first to verify the identity when the spinor $\psi$ has the simple value (\ref{hat(psi) = (1 0 0 0) in Sigma_(+)}). Next, as for Formulae (\ref{I^* dX = psi^dag gamma (dX) psi BIS}) and (\ref{J^* dX = psi^dag gamma (dX) gamma_5 psi BIS}), we prove the identity for every $\psi$ in $\Sigma_{+}$ by showing it has the \textbf{relativistic covariance}. As for  $\Sigma_{-}$, we have to proceed in the same way starting with the simple value (\ref{hat(psi) = (0 0 0 1) in Sigma_(-)}). In the sequel, we present the principle of the calculations only for $\Sigma_{+}$.

\vspace{0.5cm}

\textbf{(i) Simple case.}  Introducing the simplified notations
$$ \psi'^\dag \bm{\sigma} \, \psi' = \left[ {{\begin{array}{cc}
         \psi'^\dag \sigma_1 \, \psi'  \hfill  \\           
         \psi'^\dag \sigma_2 \, \psi'  \hfill  \\  
         \psi'^\dag \sigma_3 \, \psi'  \hfill  \\  
   \end{array} }} \right], \quad
    \psi''^\dag \bm{\sigma} \, \psi'' = \left[ {{\begin{array}{cc}
         \psi''^\dag \sigma_1 \, \psi''  \hfill  \\           
         \psi''^\dag \sigma_2 \, \psi''  \hfill  \\  
         \psi''^\dag \sigma_3 \, \psi''  \hfill  \\  
   \end{array} }} \right], \quad  
    \Im (\psi''^\dag \bm{\sigma} \, \psi' )= \left[ {{\begin{array}{cc}
         \Im (\psi''^\dag \sigma_1 \, \psi')  \hfill  \\           
         \Im (\psi''^\dag \sigma_2 \, \psi')  \hfill  \\  
         \Im (\psi''^\dag \sigma_3 \, \psi')  \hfill  \\  
   \end{array} }} \right]
$$
the relation (\ref{I^0 = & I^k =}) and (\ref{J^0 = & J^k =}) can be recast respectively into
\begin{equation}
   U = \left[ {{\begin{array}{cc}
          \mid \psi' \mid^2 + \mid \psi'' \mid^2   \hfill  \\           
         2 \, \Im (\psi'^\dag \bm{\sigma} \, \psi'')  \hfill  \\  
   \end{array} }} \right]
\label{I^0 = & I^k = BIS}
\end{equation}
\begin{equation}
   J = \left[ {{\begin{array}{cc}
          2 \, \Im (\psi'^\dag \psi'')                                                                    \hfill  \\           
          -(\psi'^\dag \bm{\sigma} \, \psi'  + \psi''^\dag \bm{\sigma} \, \psi'') \hfill  \\  
   \end{array} }} \right]
\label{J^0 = & J^k = BIS}
\end{equation}
Hence, owing to (\ref{J (Pi_1, Pi_2) =}), the expression of (\ref{Omega = Omega = J(U, J)}) is
$$  \Omega =  \left[ {{\begin{array}{cc}
        0    \hfill &  V^T    \hfill  \\
        V    \hfill &  j (W)   \hfill  \\
   \end{array} }} \right] \;\; \mbox{with} \;\; 
$$
\begin{equation}
      \left\lbrace \begin{array}{l}
    V = - 2 \, \Im (\psi''^\dag \bm{\sigma} \, \psi' ) \times (\psi'^\dag \bm{\sigma} \, \psi' + \psi''^\dag \bm{\sigma} \, \psi'' )   \\           
    W = - (\mid \psi' \mid^2 + \mid \psi'' \mid^2) \, (\psi'^\dag \bm{\sigma} \, \psi' + \psi''^\dag \bm{\sigma} \, \psi'' )  
            - 4 \,  \Im (\psi'^\dag \psi'')  \, \Im (\psi''^\dag \bm{\sigma} \, \psi' )   \\  
   \end{array} \right.
\label{V = & W  =}
\end{equation}
Taking into account (\ref{I = (1 0) & J = = (0 n) }), we know that for the simple case where the spinor value is (\ref{hat(psi) = (1 0 0 0) in Sigma_(+)}), (\ref{I^* dX = psi^dag gamma (dX) psi BIS}) and (\ref{J^* dX = psi^dag gamma (dX) gamma_5 psi BIS}) provide
\begin{equation}
    \psi = \left[ {{\begin{array}{c}
           1  \hfill  \\
           0  \hfill  \\
           0  \hfill  \\           
           0  \hfill  \\
   \end{array} }} \right], \qquad
    U = \left[ {{\begin{array}{c}
           1  \hfill  \\
           0  \hfill  \\
           0  \hfill  \\           
           0  \hfill  \\
   \end{array} }} \right], \qquad
    J = \left[ {{\begin{array}{c}
           0  \hfill  \\
           0  \hfill  \\
           0  \hfill  \\           
         -1  \hfill  \\
   \end{array} }} \right]       
\label{simple case psi = (1 0 0 0) & U = (1 0 0 0) & J = (0 0 0 1)}
\end{equation}
To avoid managing cumbersome expressions, we introduce again simplified notations
$$ z'  = z'_1 \overline{z'_2}, \quad
     z'' = z''_1 \overline{z''_2}, \quad
     \rho' = \mid \psi' \mid^2 = \mid z'_1 \mid^2 + \mid z'_2 \mid^2, \quad
     \rho'' = \mid \psi'' \mid^2 = \mid z''_1 \mid^2 + \mid z''_2 \mid^2 
$$
$$ \sigma'  = \mid z'_1 \mid^2 - \mid z'_2 \mid^2, \quad
     \sigma''  = \mid z''_1 \mid^2 - \mid z''_2 \mid^2 
$$
Then (\ref{V = & W  =}) gives
\begin{equation}
        V = \left[ {{\begin{array}{c}
           - \Re (\overline{z''_2} z'_1 - \overline{z''_1} z'_2) \, (\sigma' + \sigma'')
           -   2 \, \Im  (\overline{z''_1} z'_1 - \overline{z''_2} z'_2)  \,  \Im (z' + z'') \hfill  \\
           - 2 \, \Im  (\overline{z''_1} z'_1 - \overline{z''_2} z'_2)  \, \Re (z' + z'')  
           +  \Im (\overline{z''_2} z'_1 + \overline{z''_1} z'_2) \,  (\sigma' + \sigma'') \hfill  \\
             2 \,  \Im (\overline{z''_2} z'_1 + \overline{z''_1} z'_2) \, \Im (z' + z'') 
           + 2 \, \Re (\overline{z''_2} z'_1 - \overline{z''_1} z'_2) \, \Re  (z' + z'')      \hfill  \\           
   \end{array} }} \right]
\label{V =}
\end{equation}
\begin{equation}
     W = (\rho' + \rho'') \, 
       \left[ {{\begin{array}{c}
           - 2 \, \Re (z' + z'') \hfill  \\
            2 \, \Im (z' + z'')  \hfill  \\
           - (\sigma' + \sigma'')  \hfill  \\           
   \end{array} }} \right] 
   + 4 \, \Im  (\overline{z''_1} z'_1 + \overline{z''_2} z'_2)  \, 
       \left[ {{\begin{array}{c}
          \Im (\overline{z''_2} z'_1 + \overline{z''_1} z'_2)  \hfill  \\
          \Re (\overline{z''_2} z'_1 - \overline{z''_1} z'_2)   \hfill  \\
           \Im  (\overline{z''_1} z'_1 - \overline{z''_2} z'_2)  \hfill  \\           
   \end{array} }} \right]    
\label{W =}
\end{equation}
For the simple case (\ref{simple case psi = (1 0 0 0) & U = (1 0 0 0) & J = (0 0 0 1)}), we have
\begin{equation}
     z'_1 = 1, \qquad z'_2 = z''_1 = z''_2 = 0
\label{z'_1 = 1 & z'_2 = z''_1 = z''_2 = 0}
\end{equation}
then (\ref{V =}) and (\ref{W =}) provide
\begin{equation}
    V = \left[ {{\begin{array}{c}
           0  \hfill  \\
           0  \hfill  \\
           0  \hfill  \\           
   \end{array} }} \right], \qquad
       W = \left[ {{\begin{array}{c}
           0  \hfill  \\
           0  \hfill  \\
        - 1  \hfill  \\           
   \end{array} }} \right]   = - e_3 
\label{simple case V = & W =}
\end{equation}
where $e_3$ denotes the third vector of the canonical basis of $\mathbb{R}^3$. Now we are able to calculate the right hand side of (\ref{Im((d psi)^dag delta psi) = (1/4) Tr (d Omega Omega delta Omega)})
$$ \frac{1}{4} \, Tr (d \Omega \, \Omega \, \delta \Omega) = \frac{1}{4} \, Tr (
\left[ \begin{array}{cc}
        0    \hfill &  dV^T    \hfill  \\
        dV    \hfill &  j (dW)   \hfill  \\
   \end{array}  \right] \,
   \left[ \begin{array}{cc}
        0    \hfill &  V^T    \hfill  \\
        V    \hfill &  j (W)   \hfill  \\
   \end{array}  \right], \,
   \left[ \begin{array}{cc}
        0    \hfill &  \delta V^T    \hfill  \\
        \delta V    \hfill &  j (\delta W)   \hfill  \\
   \end{array}  \right]
)
$$
\begin{equation}
       \frac{1}{4} \, Tr (d \Omega \, \Omega \, \delta \Omega) = \frac{1}{4} \, 
W \cdot (dW \times \delta W - dV \times \delta V) + V \cdot (\delta V \times dW - dV \times \delta W)
\label{Tr (d Omega Omega delta Omega) = W cdot (dW x delta W - dV x delta V) + V cdot (delta V times dW  - dV x delta W}
\end{equation}
Differentiating (\ref{V =}) and (\ref{W =})  at the value (\ref{hat(psi) = (1 0 0 0) in Sigma_(+)}) of the spinor, we obtain for the simple case
$$     dV = \left[ {{\begin{array}{c}
           - 2 \, \Re (\overline{dz''_2})  \hfill  \\
              2 \, \Im (\overline{dz''_2})   \hfill  \\
           0  \hfill  \\           
   \end{array} }} \right], \qquad
       dW = \left[ {{\begin{array}{c}
             - 2 \, \Re (\overline{dz'_2})   \hfill  \\
                2 \, \Im (\overline{dz'_2})   \hfill  \\
             - 4 \, \Im (\overline{dz'_1})  \hfill  \\           
   \end{array} }} \right]    
$$
The expressions of $\delta V$ and $\delta W$ are deduced from the previous one replacing $d$ by $\delta$. For the  simple case, (\ref{Tr (d Omega Omega delta Omega) = W cdot (dW x delta W - dV x delta V) + V cdot (delta V times dW  - dV x delta W}) is reduced to
$$  \frac{1}{4} \,   Tr (d \Omega \, \Omega \, \delta \Omega)  
           = \frac{1}{4} \,  (e_3 \cdot ( dV \times \delta V) - e_3 \cdot (dW \times \delta W))
$$
and we obtain successively
$$  \frac{1}{4} \,  Tr (d \Omega \, \Omega \, \delta \Omega)  
             = - \lbrack Re (\overline{dz'_2}) \, \Im (\delta z'_2) 
             + \Re (\delta z'_2) \, \Im (\overline{d z'_2}) 
              - \Re (\overline{dz''_2}) \, \Im (\delta z''_2) 
             - \Re (\delta z''_2) \, \Im (\overline{d z''_2}) \rbrack
$$
\begin{equation}
     \frac{1}{4} \,  Tr (d \Omega \, \Omega \, \delta \Omega)  
             = - \Im ( \overline{dz'_2} \, \delta z'_2 - \overline{dz''_2} \, \delta z''_2) 
\label{(1/4) Tr (d Omega Omega delta Omega) = Im(bar(dz'_2) delta z'_2 - bar(dz''_2) delta z''_2)}
\end{equation}

On the other hand, we develop the expression of the left hand side of (\ref{Im((d psi)^dag delta psi) = (1/4) Tr (d Omega Omega delta Omega)}) 
\begin{equation}
     - \Im (d \psi^\dag \delta \psi) 
               =   - \Im ( \overline{dz'_1} \, \delta z'_1  + \overline{dz'_2} \, \delta z'_2 
                   - \overline{dz''_1} \, \delta z''_1 - \overline{dz''_2} \, \delta z''_2) 
\label{Im (d psi^dag delta psi) = }
\end{equation}
in which $d \psi$ and $\delta \psi$ are tangent vectors to the manifold $\Sigma_6$ at the point (\ref{hat(psi) = (1 0 0 0) in Sigma_(+)}). Then, differentiating the normalization condition (\ref{normalization condition}), we have
$$  \overline{dz'_1} \, z'_1 + \overline{z'_1} \, dz'_1 
   +  \overline{dz'_2} \, z'_2 + \overline{z'_2} \, dz'_2 
   - \overline{dz''_1} \, z'_1 + \overline{z''_1} \, dz'_1 
   -  \overline{dz''_2} \, z'_2 + \overline{z''_2} \, dz'_2 = 0
$$
that, owing to (\ref{z'_1 = 1 & z'_2 = z''_1 = z''_2 = 0}) is reduced to
$$ dz'_1  + \overline{dz'_1} = 2 \Re (\overline{dz'_1} ) = 0
$$
Likewise, we have
$$ \Re (\delta z'_1) = 0
$$
and consequently
\begin{equation}
    \Im ( \overline{dz'_1} \, \delta z'_1) = \Re (\overline{dz'_1}) \, \Im (\delta z'_1)
                                                                + \Re (\delta z'_1)\, \Im (\overline{dz'_1}) = 0
\label{Im (bar(dz'_1) delta z'_1) = 0}
\end{equation}
Similarly, differentiating the condition (\ref{psi^dag gamma_5 psi = 0}) at the point (\ref{hat(psi) = (1 0 0 0) in Sigma_(+)}), we deduce 
\begin{equation}
    \Im ( \overline{dz''_1} \, \delta z''_1) =  0
\label{Im (bar(dz''_1) delta z''_1) = 0}
\end{equation}
Simplifying (\ref{Im (d psi^dag delta psi) = }) with (\ref{Im (bar(dz'_1) delta z'_1) = 0}) and (\ref{Im (bar(dz''_1) delta z''_1) = 0}), it holds
$$      - \Im (d \psi^\dag \delta \psi) 
               =   - \Im (  \overline{dz'_2} \, \delta z'_2 
                    - \overline{dz''_2} \, \delta z''_2) 
$$
By comparison with (\ref{(1/4) Tr (d Omega Omega delta Omega) = Im(bar(dz'_2) delta z'_2 - bar(dz''_2) delta z''_2)}), we conclude that the identity (\ref{Im((d psi)^dag delta psi) = (1/4) Tr (d Omega Omega delta Omega)}) is true for the simple case where the spinor value is (\ref{hat(psi) = (1 0 0 0) in Sigma_(+)}). 

\vspace{0.5cm}

\textbf{(ii) Relativistic covariance.} Let $P$ be a Lorentz transformation and $S = S(P)$, according  to Theorem \ref{Thm S(P) for Lorents transformations}. Taking into account (\ref{transformation law of Omega}) and (\ref{tilde(psi) = S^(-1) psi = S^dag psi & gamma(tilde(U)) =}), we have
$$ Tr (d \tilde{\Omega} \, \tilde{\Omega} \, \delta \tilde{\Omega})  
   = Tr ((P^* d \Omega \, P) \, (P^* \Omega \, P)\, (P^* \delta \Omega \, P))  
   = Tr (d \Omega \, \Omega \, \delta \Omega)  
$$
$$ \Im (d \tilde{\psi}^\dag \delta \tilde{\psi}) 
      = \Im ((S^\dag d \psi)^\dag (S^\dag \delta \psi)) 
      = \Im (d \psi^\dag \delta \psi) 
$$
that achieves the proof. $\blacksquare$

\vspace{0.5cm}

By Theorem \ref{thm link psi <-> (I  J)}, we have seen that there exists a surjective map  from the manifold $ \mathcal{W}_{10} $ defined by ( \ref{W_(10) =})  to the manifold $\mathcal{V}_9$ defined by (\ref{V_9 defi}) and
\begin{equation}
     U^* dX = \psi^\dag \gamma (dX) \, \psi
\label{U^* dX = psi^dag gamma(dX) psi}
\end{equation}
of exterior derivative
\begin{equation}
     \delta U^* dX - d U^* \delta X = 2  \, \Re (\delta \psi^\dag \gamma (dX) - d \psi^\dag \gamma (\delta X)) \, \psi
\label{delta U^* dX - d U^* delta X =}
\end{equation}
Now we are able to equip $ \mathcal{W}_{10} $   with the presymplectic form for the particle of rest mass $m_0$ and spin $s$
\begin{equation}
     \omega_\mathcal{W} (\delta \eta, d\eta) = 
         4 \, s \, \Im (d \psi^\dag \delta \psi)
        - 2 \, m_0 \, \Re \lbrack \lbrace \delta \psi^\dag \gamma (dX) - d \psi^\dag \gamma (\delta X) \rbrace \, \psi \rbrack
\label{omega_W =}
\end{equation}
which, according to Theorem \ref{thm symplectic structure of Sigma_6} and (\ref{delta U^* dX - d U^* delta X =}), is the pull-back by this surjective map of the presymplectic form (\ref{Symplectic form for the relativistic particle of spin s}) on $\mathcal{V}_9$. Besides, let us notice that, differentiating (\ref{normalization condition}),
$$ \psi \in \Sigma_6 \;\; \Rightarrow \;\; 
      d (\psi^\dag \psi ) = \psi^\dag d \psi  + d\psi^\dag \psi 
      = 2 \, \Re (\psi^\dag d\psi) 
      = 0 \;\; \Rightarrow \;\;
     i \psi^\dag d \psi \in \mathbb{R}
$$
If we define on $\mathcal{W}_{10}$ the 1-form
\begin{equation}
     \alpha_\mathcal{W}  (d \eta) = - 2 \,  i \, s \, \psi^\dag d \psi - m_0 \, \psi^\dag \gamma (dX) \, \psi
\label{varpi_W =}
\end{equation}
we verify that (\ref{omega_W =}) is the exterior derivative of (\ref{varpi_W =})
$$ \omega_\mathcal{W}  = \mbox{d}\alpha_\mathcal{W} 
$$

To determine the characteristic foliation of $\alpha_\mathcal{W}$, it is more convenient and always possible by a suitable basis change of $\mathcal{E}$  to return to the simple case where
\begin{equation}
    \psi = \left[ {{\begin{array}{c}
           z'_1  \hfill  \\
           0  \hfill  \\
           0  \hfill  \\           
           0  \hfill  \\
   \end{array} }} \right], \qquad  z'_1 = e^{i \, \theta} \in \mathsf{U} (1)
\label{psi = (z'_1 0 0 0) z'_1 = e^(i theta) in U(1)}
\end{equation}
which belongs to $\Sigma_6$. We notice that 
$$ \Im (d \psi^\dag \delta \psi)  = \Im (\overline{d z'_1} \, \delta  z'_1) =  \Im (- i \, d\theta \, \overline{ z'_1} . i \,\delta \theta \,  z'_1) 
          = \Im (d\theta \, \delta \theta) = 0, \qquad 
          \psi^\dag d \psi  = \overline{d  z'_1} \, \delta  z'_1
$$
then, owing  to (\ref{delta U^* dX - d U^* delta X =}) and   (\ref{omega_W =})
$$ d \eta \in \mbox{ker} \, \omega_\mathcal{W}  \;\;  \Leftrightarrow \;\;
\forall \delta X, \delta U, \;\; m_0 \, (\delta X^* dU - dX^* \delta U) = 0 \;\;  \Leftrightarrow \;\;
dX = dU = 0
$$
and, owing to  (\ref{varpi_W =})
$$ d \eta \in \mbox{ker} \, \alpha_\mathcal{W}  \;\;  \Leftrightarrow \;\;
-2 \, i \, s \, \overline{z'_1} \, d  z'_1 - m_0 \,  U^* dX = 0  \;\;  \Leftrightarrow \;\;
z'_1 = z \, e^{\frac{i \, m_0}{2 \, s}  \, U^* X}, \;\; z = C^{te}, \;\;  z \in \mathsf{U}(1)
$$
that leads to determine the leaves of the characteristic foliation  of $\alpha_\mathcal{W}$
$$ d \eta \in \alpha_\mathcal{W}  \cap \mbox{ker} \, \omega_\mathcal{W}  \qquad  \Leftrightarrow \qquad
dX = dU = 0, \qquad d  z'_1 + m_0 \,  U^* dX = 0
$$
from which we deduce
$$X  = C^{te}, \;\;   U = C^{te}, \;\;  \psi = e^{\frac{i \, m_0}{2 \, s}  \, U^* X} \zeta, \;\;
\zeta = \left[ {{\begin{array}{c}
           z  \hfill  \\
           0  \hfill  \\
           0  \hfill  \\           
           0  \hfill  \\
   \end{array} }} \right], \;\; z = C^{te}, \;\;  z \in \mathsf{U}(1)
$$
hence the 8 coordinates of the unique intersection point of the leaf and the transverse manifold $\mathcal{Y}_9$ are the 4 coordinates of $X \in E_{1,3}$, the 3 independent coordinates of $U \in E_{1,3}$ such that $U^* U = 1$ and the unique coordinate of $z \in \mathsf{U}(1)$. 
This foliation is sectionable and then, if we define for this action 
$$ \xi = \Phi (\eta) = \mbox{orb} (\eta)
$$
$\xi$ varies over a manifold $\mathcal{Y}_9$ in such a way that $\Phi$ is a smooth map from $\mathcal{W}_{10}$ to $\mathcal{Y}_9$. The 1-form $\alpha_\mathcal{W}$, which  is an integral invariant, pases to the quotient $\mathcal{Y}_9$ and defines a 1-form  $\alpha$ on $\mathcal{Y}_9$. It equips $\mathcal{Y}_9$ with a structure of a prequantum manifold  and $(\mathcal{Y}_9, \pi)$ is a prequantization of the space of motions $\mathcal{U}_8$ (Figure \ref{fig Prequantization of the particle with spin with spin 1/2}). 

\vspace{0.5cm}

We prove by Theorem \ref{thm action conservation)} that the value of the action integral does not depends on the choice of the orbit of $\mathsf{U}(1)$.  For easiness, we calculate that for the orbit passing by (\ref{hat(psi) = (1 0 0 0) in Sigma_(+)}) of which the elements are of the form (\ref{psi = (z'_1 0 0 0) z'_1 = e^(i theta) in U(1)})
$$ \int_c \alpha (d \xi)   = - 2 \,  i \, s \,  \int_c \, \psi^\dag d \psi 
                                      = - 2 \,  i \, s \,  \int_c \,  \overline{z'_1} \, d  z'_1
                                      = 2 \,  s \,  \int^{2 \, \pi}_0 \,  d \theta
$$
Then the Sommerfeld quantization condition (\ref{Sommerfeld quantization condition BIS}) gives
$$     \int_c \alpha (d \xi)   = 4 \, \pi \, s = 2 \, \pi \, n, \qquad n \in \mathbb{N}
$$
and the spin must be an integer multiple of $1/2$
$$ s = \frac{n}{2}, \qquad n \in \mathbb{N}
$$
Later on, we consider only the case of the electron for which $s = 1/2$.

\section{Geometric quantization}
\label{Section - Geometric quantization}

In  \cite{SSD, SSDEng}, a \textbf{Planck manifold} is defined to be any submanifold $\mathcal{F}$ of the prequantum manifold $\mathcal{Y}$ on which the 1-form $\alpha$ vanishes
$$ \lbrack \delta \xi \;\; \mbox{tangent to} \;\;  \mathcal{F} \,\; \mbox{at} \;\; \xi \rbrack \;\; \Rightarrow \;\;
\alpha (\delta \xi) = 0
$$
then every Planck manifold is isotropic with respect to $\omega_\mathcal{Y}$. If $\pi$ is injective on $\mathcal{F}$, the image $F = \pi (\mathcal{F})$ is  diffeomorphic to $\mathcal{F}$, then $F$ is isotropic with respect to $\omega_\mathcal{U}$. We say that $\mathcal{F}$ is a Planck lift of $F$. Conversely, if we know an isotropic  foliation $x \mapsto H$ of $\mathcal{U}$, we define a foliation $\xi \mapsto \mathcal{H}$ of $\mathcal{Y}$, called \textbf{Planck foliation}. The leaves of $\mathcal{H}$ are, at least locally, Planck lifts of the leaves of $H$.

Let $u: \mathcal{U} \mapsto \mathbb{R}$ be a smooth field on the space of motions (which Souriau call a  dynamical variable). Using Souriau's notations, there exists a tangent vector field
$$\underset{u}{\delta} x = \mbox{grad} \, u
$$
called the \textbf{symplectic gradient} by Souriau (or also Hamiltonian vector field) such that
\begin{equation}
     \iota_{\underset{u}{\delta} x} \, \omega_\mathcal{U}  = - du
\label{symplectic gradient defi}
\end{equation}
for any scalar- or vector-valued field $f$, it defines a derivation
$$ \underset{u}{\delta} f = (\underset{u}{\delta} x )(f)
$$
We call \textbf{state vector} every  smooth function $\Psi$ on $\mathcal{Y}$ with compact support such that
\begin{equation}
     \Psi (z \cdot \xi) = z \, \Psi (\xi), \quad z \in \mathsf{U}(1)
\label{Psi (z cdot xi) = z Psi (xi)}
\end{equation}
then  $\Psi^\dag (\xi)  \Psi' (\xi) $ is constant on the  orbit  of  $\mathsf{U}(1)$ and we may apply it the measure $\mbox{d}\mu$  on $\mathcal{U}$. As usual in quantum mechanics, we consider a set  $\mathcal{H} (\mathcal{Y})$ of such functions which is a Hilbert space for the complex inner product 
$$\left< \Psi, \Psi' \right> = \int_\mathcal{U} \Psi^\dag \Psi' \, \mbox{d}\mu
$$
To every dynamical variable defined on $\mathcal{U}$, we can associate an Hermitian operator $\hat{u}$ on $\mathcal{H} (\mathcal{Y})$, called \textbf{observable}, defined by
$$ \hat{u} \, \Psi = - i \, \underset{u}{\delta} \, \Psi
$$
such as the map $u \mapsto \hat{u} $ is $\mathbb{R}$-linear  and injective and
$$ \hat{1} = 1_{\mathcal{H} (\mathcal{Y})}, \qquad 
     \hat{u} \, \hat{u'} - \hat{u'} \, \hat{u} = - i \, \widehat{\left\lbrace u, u' \right\rbrace}
$$
then 
$$ \underset{1}{\delta} \, \Psi = i \, \hat{1} \, \Psi = i \, \Psi
$$

In (\cite{SSD, SSDEng}, (18.71)), a construction is proposed to determine the state vector which described what is happening during an experiment, based on an isotropic foliation of the space of motions $\mathcal{U}$ which, in favorable circumstances, can be lifted to a Planck foliation. Let us assume that the dynamical variables $u_j \, (j = 1, 2, \ldots , p)$ are in involutions ($\left\lbrace u_j , u_k \right\rbrace = 0$). If the equations $u_j = C^{te}$ define connected submanifolds of the space of motions, then these manifolds are leaves of a coisotropic foliation. At each $x \in \mathcal{U}$, the vector space orthogonal to the leaf passing by $x$ is the kernel of the 2-form induced on the leaf. Thus on each coisotropic leaf this kernel determines an isotropic foliation.  The corresponding Planck lift is the foliation $\xi \mapsto \mathcal{H}$ of $\mathcal{Y}$. $\mathcal{H}$ is the vector space spanned by the vectors
$$ \underset{u_j}{\delta} \xi - u_j \, \underset{1}{\delta} \xi
$$
The corresponding \textbf{Planck condition} will be satisfied by a state vector if
\begin{equation}
    \underset{u_j}{\delta} \, \Psi = u_j \, \underset{1}{\delta} \, \Psi = i \, u_j \, \Psi
\label{Planck condition}
\end{equation}
Let $G$ be a symplectic group of the space of motions $\mathcal{U}$ with a null class of symplectic cohomology and $\tilde{G}$ be a normal Abelian Lie subgroup of $G$. Then it is shown in (\cite{SSD, SSDEng}, (19.39)) that if we choose a basis $Z_1, \ldots , Z_p$ for $\tilde{\mathfrak{g}}$, then the dynamical variables $u_j = \mu (Z_j)$ are in involutions. 

\section{Quantization of the relativistic particle with  spin 1/2}
\label{Section - Quantization of the relativistic particle with  spin 1/2}

We can apply the previous construction and determine the corresponding Planck condition by taking for $G$ the Poincar\'e group and for $\tilde{G}$ the group of space-time translations of momentum $\mu$ given by (\ref{mu (Z) = - Pi delta C - (1/2) Tr (M delta P)}) and (\ref{ Pi = m I & I^* I = 1 & J^* J = -1 & J^* I = 0})
$$ \mu (Z) = - \Pi^* \delta X  = - m_0 \,  U^* \delta X 
$$
The dynamical variables $u_j$ are the energy $\Pi^0$ and the components $\Pi^k \, (k = 1, 2, 3)$ of the linear momentum. Owing to (\ref{Symplectic form for the relativistic particle of spin s}) and (\ref{symplectic gradient defi}), the symplectic gradient of $\mu(Z)$ is given by
$$ \iota_{\underset{\mu(Z)}{\delta} x} \, \omega_\mathcal{U} (dx)
    = \omega_\mathcal{U} (\underset{\mu(Z)}{\delta} x, dx)
    = m_0 (\underset{\mu(Z)}{\delta} X^* dU - dX^* \underset{\mu(Z)}{\delta} U)
    =  - \underset{\mu(Z)}{\delta} X^* (- m_0 \, dU)
$$
then 
$$ \underset{\mu(Z)}{\delta} x = \delta X^* \, \frac{\partial}{\partial X} 
         = \delta X^j \partial_j
$$
The Planck condition (\ref{Planck condition}) leads to the partial derivative system
$$ \frac{\partial \Psi}{\partial X} = - i \, m_0 \, U^* \, \Psi
$$
of which the solution in  $\mathcal{H} (\mathcal{Y})$  are of the form
$$ \Psi (X, U, z) = e^{- i \, m_0 \, U^* X} \, \psi (U, z) 
$$
The condition (\ref{Psi (z cdot xi) = z Psi (xi)}) gives
$$ \psi (U, z) = \psi_0 (U) \, z, \qquad z \in \mathsf{U}(1)
$$
In the sequel, the shall not write explictly the dependence of the state vector with respect to $z$ and write
$$ \Psi (\xi) = e^{- i \, m_0 \, U^* X} \, \psi (U) 
$$
If the smooth function $U \mapsto \psi(U)$ on the manifold 
$$\mathcal{M} = \lbrace U \in \mathcal{E}_{1,3} \;\; \mbox{such that} \; \; 
                         U^* U = 1 \;\; \mbox{and} \;\; U \; \mbox{future-directed}
\rbrace
$$ 
equipped with the measure  $\mbox{d}\mu$, has compact support, we can define 
\begin{equation}
     \tilde{\Psi} (X) = \int_\mathcal{M} e^{- i \, m_0 \, U^* X}  \, \psi \, \mbox{d}\mu
\label{tilde(Psi) (X) = int_H e^(-i m_0 U^*X d mu}
\end{equation}
where $\psi$ is an eigenvector of $\gamma(U)$ with eigenvalue $\epsilon$.
Owing to $U^* X = G_{jk} U^j X^k$, (\ref{gamma^j = G^(jm) gamma_m}) and  (\ref{gamma : E_(1,3) -> Gamma_(+) : X -> gamma(X) = X_i gamma_i}), it holds
$$ i \, \gamma^k \partial_k  \tilde{\Psi} = - i \, G^{kl} \gamma_l
      \int_\mathcal{M} \, i \, m_0 \, G_{jk} U^j  \,  e^{- i \, m_0 \, U^* X} \, \psi \,  \mbox{d}\mu
      =  m_ 0 \int_\mathcal{M}  \,  e^{- i \, m_0 \, U^* X} \, \gamma (U) \, \psi \,   \mbox{d}\mu
$$
that leads to the \textbf{Dirac equation}, owing to  (\ref{gamma (U)  psi = epsilon psi})
\begin{equation}
     i \, \gamma^k \partial_k  \tilde{\Psi} - \epsilon \, m_0 \, \tilde{\Psi} = 0
\label{Dirac equation}
\end{equation}

Let us define 
$$ \tilde{U}^j = \tilde{\Psi}^\dag \gamma^j \tilde{\Psi}
$$
Owing to (\ref{I^0 = & I^k =})  and (\ref{I = U & U^* U = 1 & J^*J = -1 & J^*U = 0}), we have
$$  \tilde{U}^0 = \mid \tilde{\Psi}' \mid^2 + \mid \tilde{\Psi}'' \mid^2  \geq 0
$$
then $\tilde{U}^0$ being non negative, it can be interpreted as a \textbf{probability density}. It holds
$$ \partial_j \tilde{U}^j = \tilde{\Psi}^\dag \gamma^j \, \partial_j \tilde{\Psi}
                   + (\partial_j \tilde{\Psi}^\dag) \, \gamma^j \, \tilde{\Psi}
$$
As the matrices $\gamma^j$ are Hermitian,
$$ \partial_j \tilde{U}^j = \tilde{\Psi}^\dag (\gamma^j \, \partial_j \tilde{\Psi})
                   + (\gamma^j \, \partial_j \tilde{\Psi})^\dag \, \tilde{\Psi}
$$
and, taking into account the Dirac equation (\ref{Dirac equation})
$$  \partial_j \tilde{U}^j =0
$$
which proves the \textbf{conservation of the probability current} $\tilde{U}$. 

\vspace{0.5cm}

In particular, if $\mathsf{d} \mu$ is a Dirac measure, it is worth to observe that we have
\begin{equation}
    \tilde{\Psi} = e^{- i \, m_0 \, U^* X}  \, \psi 
\label{tilde(Psi) = e^() i m_0 U^*X) psi}
\end{equation}
Taking into account (\ref{I^j = psi^dag gamma^j psi}) and (\ref{I = U & U^* U = 1 & J^*J = -1 & J^*U = 0}), we have
$$ \tilde{U}^j = \psi^\dag \gamma^j \psi = U^j 
$$
that establishes a link between this current and the 4-velocity of the corresponding classical particle. 

\vspace{0.5cm}

Moreover, it is worth to have a look to the 4-vector $J$ defined by (\ref{J^* dX = psi^dag gamma (dX) gamma_5 psi}). It is easy to verify  that, owing to (\ref{gamma^j = G^(jm) gamma_m}), its components are 
$$J^k = i \, \psi^\dag \gamma^k \gamma_5 \, \psi
$$
By analogy, we introduce the 4-vector $\tilde{J}$ of components
$$ \tilde{J}^k  = i \, \tilde{\Psi}^\dag \gamma^k \gamma_5 \, \tilde{\Psi}
$$
Owing to (\ref{gamma(U) gamma(V) + gamma(V) gamma(U) = (U^* V) 1_(C^4)}), we have
\begin{equation}
     \frac{1}{2} \, \lbrack \gamma^k \gamma_l + \gamma_l \, \gamma^k \rbrack = \delta^k_l
\label{(1/2) ( gamma^k gamma_l + gamma_l gamma^k) = delta^k_l}
\end{equation}
and it holds
$$ \partial_k \tilde{J}^k = i \, \tilde{\Psi}^\dag \gamma^k \gamma_5 \, \partial_k \tilde{\Psi}
                           + i \, (\partial_k \tilde{\Psi}^\dag) \, \gamma^k \gamma_5 \, \tilde{\Psi}
$$ 
Transforming the first term with (\ref{(1/2) ( gamma^k gamma_l + gamma_l gamma^k) = delta^k_l}) and the second term thanks to the fact that the matrix $\gamma^k$ is Hermitian, we have
$$ \partial_k \tilde{J}^k = - \tilde{\Psi}^\dag \gamma_5 \,  (i \, \gamma^k \, \partial_k \tilde{\Psi})
                           - ( i \, \gamma^k \partial_k \tilde{\Psi})^\dag  \,  \gamma_5 \, \tilde{\Psi}
$$
Owing to the Dirac equation (\ref{Dirac equation}), we obtain
$$ \partial_k \tilde{J}^k =- 2 \epsilon \, m_0 \,  \tilde{\Psi}^\dag \gamma_5 \,   \tilde{\Psi}
$$
In the case where $\tilde{\Psi}$ is given by (\ref{tilde(Psi) = e^() i m_0 U^*X) psi}), we have
$$ \tilde{\Psi}^\dag \gamma_5 \,  \tilde{\Psi}  = \psi^\dag \gamma_5 \,   \psi
$$
As $\psi \in \Sigma_6$, we have
$$ \tilde{\Psi}^\dag \gamma_5 \,  \tilde{\Psi} = 0
$$
from which we deduce the \textbf{conservation} of the  \textbf{polarization current} or \textbf{spin current} $\tilde{J}$
$$ \partial_k \tilde{J}^k = 0
$$

\section{Charge conjugation, parity transformation and time reversal symmetries}
\label{Section - Charge conjugation, parity transformation and time reversal symmetries}

To represent theses symmetries in the present framework, we study the properties of the orthonormal basis (\ref{defi gamma_alpha}). Introducing
$$ \tau_j = \left\lbrace \begin{array}{l}
        +1 \; \mbox{if} \; j = 0  \\
        - 1 \; \mbox{otherwise}   \\
   \end{array}
  \right. 
$$
the index rising (\ref{gamma^j = G^(jm) gamma_m}) for the metric (\ref{X^* X = metric of E_(1,4)}) 
$$      X^* X =  (X^0)^2 - (X^1)^2 - (X^2)^2 -(X^3)^2 - (X^5)^2
                   =  t^2 - (x_1)^2 - (x_2)^2- (x_3)^2- y^2
$$
of the 5D hyperbolic Euclidean space $E_{1,4}$ becomes
\begin{equation}
     \gamma^j = \tau_j \gamma_j
\label{gamma^j = tau_j gamma_j}
\end{equation}
Owing to (\ref{(gamma_alpha)^2 = for alpha = 0...5}), we have
\begin{equation}
      (\gamma_j)^{-1} = \tau_j \gamma_j
\label{(gamma_j)^(-1) = tau_j gamma_j}
\end{equation}
and, because of  the hermiticity of the matrices $\gamma_j$,
$$ (\gamma_j)^{-1} = \tau_j \gamma^\dag_j
$$
the matrix $\gamma_0$ is \textbf{unitary} while the matrices $\gamma_1, \gamma_2, \gamma_3, \gamma_4$ are \textbf{neg-unitary}, that is
$$ \gamma_j \gamma^\dag_j = - 1_{\mathbb{C}^4} \qquad  \mbox{for} \; j \in \lbrace 1, 2, 3, 5 \rbrace
$$

Taking into account (\ref{gamma^j = tau_j gamma_j})  and  (\ref{(gamma_j)^(-1) = tau_j gamma_j}), we have
$$ (\gamma_k)^{-1} \gamma^k \gamma_k  = (\gamma_k)^2 \gamma_k  = \tau_k \gamma_k 
$$
then
\begin{equation}
      (\gamma_k)^{-1} \gamma^k \gamma_k  = \gamma^k 
\label{(gamma_k)^(-1) gamma^k gamma_k = gamma^k}
\end{equation}
Besides, owing to (\ref{gamma^j = tau_j gamma_j}), (\ref{(gamma_j)^(-1) = tau_j gamma_j}) and (\ref{gamma_alpha gamma_beta + gamma_beta gamma_alpha = 0 for alpha neq beta}), we have for $j \neq k$
$$ (\gamma_j)^{-1} \gamma^k \gamma_j  = \tau_j \tau_k \gamma_j  (\gamma_k \gamma_j )
               = - \tau_j \tau_k (\gamma_j )^2 \gamma_k 
               = - \tau_k \gamma_k 
$$
then, owing to (\ref{gamma^j = tau_j gamma_j}), it holds
\begin{equation}
      (\gamma_j)^{-1} \gamma^k \gamma_j  = - \gamma^k \;\; \mbox{for} \; \; j \neq k
\label{(gamma_j)^(-1) gamma^k gamma_j = - gamma^k}
\end{equation}

By the choice of a suitable orthonormal basis of $E_{1,4}$,  as columns $\psi$ represent spinors of the Hermitian space $\mathcal{E}$, when we apply the transformation matrix $S$ to  the column $\tilde{\psi}$, according to 
\begin{equation}
      \psi = S^{-1}\tilde{\psi}
\label{psi = S^(-1) tilde(psi)}
\end{equation}
which means the action of a symmetry on the space $\mathcal{E}$ of spinors, $S^{-1}$ acts by conjugation on the matrices $\gamma^k$ representing elements of $End(\mathcal{E})$
\begin{equation}
      \gamma^k = S^{-1} \tilde{\gamma}^k \, S
\label{gamma^k = S^(-1) tilde(gamma)^k S}
\end{equation}
which means the corresponding action of the symmetry on the 5D space $E_{1,4}$. For $S = \gamma_j$, owing to (\ref{(gamma_k)^(-1) gamma^k gamma_k = gamma^k}) and (\ref{(gamma_j)^(-1) gamma^k gamma_j = - gamma^k}), the action on $\gamma^k$ is the identity if $j = k$ or a sign change if $j \neq k$.
Through (\ref{(gamma(X))^2 = (X^* X ) 1_(C^4)}), there is a correspondance between the matrix $\gamma_k$ and the coordinate $X^k$, up to the sign. Then we claim the following rule:

\vspace {0.2cm}

\textbf{When a symmetry acts on the space of spinor by} $\gamma_j$\textbf{ in such a way that}
$$ \gamma^k = \left\lbrace \begin{array}{l}
           \tilde{\gamma}^k \quad\;  \mbox{if} \; k = j  \\
        - \tilde{\gamma}^k \;\; \mbox{otherwise}   \\
   \end{array}
  \right. 
$$
\textbf{it acts on the 5D space by}
$$ X^k = \left\lbrace \begin{array}{l}
        - \tilde{X}^k \;\;  \mbox{if} \; k = j  \\
           \tilde{X}^k \quad\; \mbox{otherwise}   \\
   \end{array}
  \right. 
$$

\vspace {0.2cm}

For instance, with the notations of Section \ref{Section - Dirac spinors}, the symmetry acting on the spinors by $\gamma_0$ acts on the coordinates of the 5D space $E_{1,4}$ by
$$ t = - \tilde{t}, \qquad  x = \tilde{x}, \qquad  y = \tilde{y}
$$
which means the \textbf{time reversal}. We denote it $\mathsf{T}$. 

\vspace {0.2cm}

Likewise, the symmetry acting on the spinors by $\gamma_5$ acts on the coordinates by
$$ t = \tilde{t}, \qquad  x = \tilde{x}, \qquad  y = - \tilde{y}
$$
\textbf{Our hypothesis is that the  fifth coordinate }$y$ \textbf{is that of the Kaluza-Klein theory} (see \cite{de Saxce 2025}) in which the electric charge $q$ is the linear momentum with respect to $y$, then the symmetry changes its sign
$$ \tilde{q}  = m_0 \, \frac{d \tilde{y}}{d \tilde{t}} = - m_0 \, \frac{d y}{d t} = - q
$$
which means it is the \textbf{charge conjugation}. We denote it $\mathsf{C}$.

Morover, the state vector $\tilde{\Psi} $ of an electron belongs to  $\Sigma_{+}$. When the electron is at rest,  owing to (\ref{(gamma_j)^(-1) = tau_j gamma_j}) and the transformation law of spinors (\ref{psi = S^(-1) tilde(psi)}) with $S = \gamma_5$, we have
\begin{equation}
     \Psi =  (\gamma_5)^{-1} \, \tilde{\Psi} 
            = \left[ {{\begin{array}{cc}
        0  \hfill &  1_{\mathbb{C}^2} \hfill  \\
       - 1_{\mathbb{C}^2} \hfill &  0    \hfill  \\
   \end{array} }} \right] \,
   \left[ {{\begin{array}{c}
        \psi'\hfill  \\
        0    \hfill  \\
   \end{array} }} \right]
    = \left[ {{\begin{array}{c}
        0 \hfill  \\
        - \psi'    \hfill  \\
   \end{array} }} \right] \in \Sigma_{-}
\label{Psi = (gamma_5)^(-1) tilde(Psi)}
\end{equation}
and the wave function $\tilde{\Psi} $ is transformed into that $\Psi$ of the positron. Because of the relativistic covariance, it is true even if these particles are not at rest. As the electron satisfies Dirac equation (\ref{Dirac equation}) with $\epsilon = 1$ and the dummy index $k$ running from 0 to 3
\begin{equation}
      \lbrack i \, \gamma^k \partial_k - m_0 \, 1_{\mathbb{C}^4} \rbrack \, \tilde{\Psi} = 0
\label{(i gamma^k partial_k - m_0 1_(C^4) ) tilde(Psi) = 0}
\end{equation}
we have
$$  (\gamma_5)^{-1} \, \lbrack i \, \gamma^k \partial_k - m_0 \, 1_{\mathbb{C}^4} \rbrack \, \tilde{\Psi} = 0
$$
or
$$  (\gamma_5)^{-1} \, \lbrack i \, \gamma^k \partial_k - m_0 \, 1_{\mathbb{C}^4} \rbrack 
      \, \gamma_5 \, (\gamma_5)^{-1}  \, \tilde{\Psi} = 0
$$
Owing to (\ref{(gamma_j)^(-1) = tau_j gamma_j}), (\ref{(gamma_j)^(-1) gamma^k gamma_j = - gamma^k}) and (\ref{Psi = (gamma_5)^(-1) tilde(Psi)}), we obtain
$$ \lbrack i \, \gamma^k \partial_k + m_0 \, 1_{\mathbb{C}^4} \rbrack \, \Psi = 0
$$
that is the Dirac equation with $\epsilon = -1$.  Then the standard form (\ref{Dirac equation}) of the Dirac equation is preserved by the charge conjugation $\mathsf{C}$. 


It is easy to verify that 
$$ \gamma_0 \gamma_1 \gamma_2 \gamma_3 \gamma_5 = - 1_{\mathbb{C}^4}
$$
or, owing to (\ref{(gamma_alpha)^2 = for alpha = 0...5})
$$ \gamma_1 \gamma_2 \gamma_3 =  \gamma_0 \gamma_5
$$
Applying the aforementioned rule,  it acts on the coordinates of the 5D space $E_{1,4}$ by 
$$ t =  \tilde{t}, \qquad  x = - \tilde{x}, \qquad  y = \tilde{y}
$$
which means the \textbf{parity transformation}. We denote it $\mathsf{P}$.

\vspace{0.5cm}

\textbf{Remark 1.} It is  worth to  notice that the representation of $\mathsf{T}$ in the space of spinors is unitary while those of $\mathsf{C}$ and $\mathsf{P}$ are neg-unitary.

\vspace{0.5cm}

\textbf{Remark 2.} The previous discussion is summarized in Table 3. This systematic construction of the symmetries of charge conjugation, parity transformation and time reversal seems to us simpler and more readable than that of the classical presentations as in \cite{Bjorken 1964}.

\begin{center}
\begin{tabular}{|| c || c || c c c ||}
 \hline \hline
  symmetry &  & $\mathsf{C}$ & $\mathsf{P}$ & $\mathsf{T}$ \\
  \hline \hline
   & & & & \\
   action on spinors & $\psi$ & $\gamma_5$ & $\gamma_1 \gamma_2 \gamma_3 $ & $\gamma_0$\\
  & & & & \\ 
   \hline\hline
   & & & & \\
   & $ t = $ & $\tilde{t}$ & $\tilde{t}$  & $- \tilde{t}$ \\
   & & & & \\
action on coordinates    & $ x= $ & $\tilde{x}$ & $- \tilde{x}$  & $\tilde{x}$ \\
  & & & & \\
      & $ y= $ & $- \tilde{y}$ & $ \tilde{y}$  & $\tilde{y}$ \\
  & & & & \\
    \hline \hline
\end{tabular}
\end{center}
\begin{center}
\textbf{Table 3:} action of $\mathsf{C}, \mathsf{P}, \mathsf{T}$ on spinors and coordinates
\end{center}

\section{Conclusions}

In classical mechanics, an electron is considered as a particle while in quantum mechanics it exhibits particle or wave properties according to the experimental circumstances. In this paper, we followed Souriau's approach, the merit of which is to propose a new way to deduce the wave equation of the quantum particle from the symplectic representation of the motion of the classical particle. The main developments are algebraic from which the Dirac wave equation is deduced  only in the last section. In particular, the construction of the spin group is purely algebraic without resorting to the matrix exponential. This approach sheds light on the link between the velocity of the classical particle and the probability current of the quantum particle. 

\vspace{0.5cm}

In Souriau's geometric quantization of the electron, there are four key formulae: 
\begin{itemize}
\item The definition (\ref{I^* dX = psi^dag gamma (dX) psi}) of the 4-vector $I$
$$    \forall dX \in \mathbb{R}^4, \qquad I^* dX = \psi^\dag \gamma (dX) \, \psi 
$$
which turned out to be the 4-velocity $U$ of the classical particles and later the model of the probability current of the quantum particle that  satisfies a well-known conservation identity.
\item The definition (\ref{J^* dX = psi^dag gamma (dX) gamma_5 psi}) of the 4-vector $J$
$$       \forall dX \in \mathbb{R}^4, \qquad J^* dX = i \,\psi^\dag \gamma (dX) \, \gamma_5 \,  \psi 
$$
which proved to be the 4-spin vector of the classical particle and next a spin current of the quantum particle that satisfies also a conservation identity.
\item The link between the spin part of the symplectic form of the classical particle and the quantum counterpart in terms of spinors in the form of the identity (\ref{Im((d psi)^dag delta psi) = (1/4) Tr (d Omega Omega delta Omega)})
$$   - \Im (d \psi^\dag \delta \psi) = \frac{1}{4} \, Tr (d \Omega \, \Omega \, \delta \Omega)
$$
\item The contact structure defined by the 1-form (\ref{varpi_W =})
$$ \alpha_\mathcal{W}  (d \eta) = - 2 \,  i \, s \, \psi^\dag d \psi - m_0 \, \psi^\dag \gamma (dX) \, \psi
$$
\end{itemize}
To throw some light on them, we supported the two definitons on the basis of the decomposition (\ref{Gamma = Gamma-i oplus Gamma_(+) oplus Gamma_(-)})  of the space $\Gamma$ of quaternionic matrices and we proved two main results, Theorems \ref{thm link psi <-> (I  J)} and \ref{thm symplectic structure of Sigma_6}.

\vspace{0.5cm}

To characterize the classical particles, Jean-Marie Souriau proposes in (\cite{SSD, SSDEng}, (13.1)) axioms of mechanics, the first of these being: "the space of motions of a dynamical system is a connected symplectic manifold". He does not give explicitly a similar formal definition for the quantum particle but he represents the quantum electron with the set of spinors such that $\psi^\dag \psi = 1$, which seems to indicate that he considers the space of a quantum particle is also connected. However Dirac claims that his equation allows to represent two distinct particles, the electron and the positron. In the present formalism, we considered also the set of spinors such that $\psi^\dag \psi = -1$ which can be associated to the positron. Therefore, there are two distinct quantum particles in correspondance with a single classical particle but, unlike Dirac, these two particles have the same spin and rest mass (not of opposite sign). It is the sign of the electric charge and not that of the rest mass that distinguishes the positron from the electron. Our opinion is that this ambiguity in Dirac theory  on the sign of the rest mass (then of the energy) results from the absence of the charge in the formalism, ambiguity which might be resolved by introducing the Kaluza-Klein fifth dimension of the Universe that allows to identify naturally the charge as an extra invariant of the coadjoint orbit as proposed in \cite{de Saxce 2025}. It is also worth noting that it is troubling to observe that the space $\Gamma_{+}$ of Dirac matrices is of dimension 5 and is only used in a reduced version of dimension 4. All these considerations could be a topic of investigation which we hope to tackle in the future.

\vspace{0.5cm}

\textbf{Acknowledgement}

\vspace{0.5cm}

The author would like to thank to Richard Kerner for his suggestion to study this topics and his valuable advice.


\end{document}